\documentclass[journal,onecolumn,12pt,draftclsnofoot]{IEEEtran}

\usepackage{graphicx}
\usepackage{subfigure}
\usepackage{amstext}
\usepackage{amsthm,amsmath,amssymb}
\usepackage{bm} 
\usepackage{amsfonts}
\usepackage{microtype}

\usepackage{mathptmx}
\usepackage{amsthm,amsmath,amssymb,amsfonts}
\usepackage{bm} 
\usepackage{subeqnarray}
\usepackage{cases,cite}
\usepackage[utf8]{inputenc}
\usepackage{float}
\usepackage{graphicx} 
\usepackage{subfigure}
\usepackage{url}
\usepackage{multirow}
\usepackage{microtype}
\usepackage{booktabs}

\newtheorem{proposition}{Proposition}
\newtheorem{theorem}{Theorem}
\newtheorem{lemma}{Lemma}
\newtheorem{remark}{Remark}
\newtheorem{example}{Example}

\IEEEoverridecommandlockouts

\begin{document}

\title{Optimal Cyber-Insurance Contract Design for Dynamic Risk Management and Mitigation}

\author{Rui~Zhang, Quanyan~Zhu

\thanks{R. Zhang and Q. Zhu are with the Department of Electrical and Computer Engineering, New York University, Brooklyn, NY, 11201
E-mail:\{rz885,qz494\}@nyu.edu. }}

\maketitle

\begin{abstract}
With the recent growing number of cyberattacks and the constant lack of effective defense methods, cyber risks become ubiquitous in enterprise networks, manufacturing plants, and government computer systems. Cyber-insurance provides a valuable approach to transfer the cyber risks to insurance companies and further improve the security status of the insured. The designation of effective cyber-insurance contracts requires the considerations from both the insurance market and the dynamic properties of the cyber risks. To capture the interactions between the users and the insurers, we present a dynamic moral-hazard type of principal-agent model incorporated with Markov decision processes, which are used to capture the dynamics and correlations of the cyber risks as well as the user's decisions on the protections. We study and fully analyze a case with a two-state two-action user under linear coverage insurance, and we further show the risk compensation, Peltzman effect, linear insurance contract principle, and zero-operating profit principle in this case. Numerical experiments are provided to verify our conclusions and further extend to cases of a four-state three-action user under linear coverage insurance and threshold coverage insurance.
\end{abstract}


\section{Introduction}
Cyber risks created by malicious attackers such as ransomware \cite{o2012ransomware}, data breaches \cite{romanosky2014empirical}, and denial-of-service \cite{apiecionek2014protection}, have become severe threats to the security of important devices and private data in Internet of things (IoT) and cyber-physical systems (CPS) \cite{manshaei2013game}. For example, the CryptoLocker ransomware attack has caused an estimated loss of $\$$3 million \cite{kelion2013cryptolocker}. The 2016 Dyn cyberattack has resulted in the disruption of major Internet platforms and services to large swathes of users in Europe and North America \cite{hilton2016dyn}.

Although various defense methods such as firewalls \cite{cheswick2003firewalls}, intrusion detection systems \cite{rowland2002intrusion}, and moving-target defenses \cite{jajodia2012moving}, have been deployed to detect the intrusion attempts and protect the networked devices, they cannot eliminate the cyber risks due to the complexities of cyber-environments \cite{kumar2006managing}.  Moreover, cyber threats are becoming stealthier, more strategic and purposeful as exemplified by the advanced persistent threats such as Stuxnet attacks on Iranian nuclear power plant in 2009 and the Ukrainian power plant attack in 2015 \cite{farwell2011stuxnet,beelitz2012using}. 

Recently emerged cyber-insurance provides an economically viable solution to further mitigate the cyber risks and improve network resiliency  \cite{kesan2005cyberinsurance,bohme2010modeling,shetty2010competitive,pal2014will,zhang2017bi}. The insured network users could quickly recover from severe cyber-incidents since part of the losses have been covered by the insurers. However, like the classic insurance, the insurers may suffer from offering coverage to reckless users due to the information asymmetry that the insurers cannot directly observe the users' protections \cite{rothschild1978equilibrium,holmstrom1979moral,holmstrom1982moral}. 

Moreover, as suggested by the theory of risk compensation in traditional insurance scenarios  \cite{ewold1991insurance}, the users may become less careful against cyberattacks knowing that insurers will cover their losses, for example, users may click more phishing emails, ignore the warnings of upgrading firewalls or systems, and reduce the frequency of scanning viruses or worms. As a result, the users may encounter more severe cyber-incidents and the insurers may bear extra cyber risks.    
 
Thus, it is imperative to study cyber-insurance contracts and its impacts on the users' cyber-risk statuses. However, classic risk analysis and insurance frameworks cannot be directly applied to cyber risks and cyber-insurance as cyber risks are dynamically evolving and strongly correlated \cite{xu2014cybersecurity,fava2007terrain,cheung2003modeling,bohme2006models}. For example, an adversary can first launch a node capture attack to compromise the system  \cite{tague2008modeling,tague2009mitigation}, and then gain the administration to the devices  \cite{jones2000computer}, steal private information  \cite{romanosky2011data}, or inject Ransomware worms or viruses  \cite{gazet2010comparative}.

In this paper, we capture the correlations and dynamics of the cyber risks as well as the users' decisions on the protections with the Markov decision processes (MDP)  \cite{puterman2014markov,roy2010survey}. Different states of the MDP are used to capture the different cyber risks from various sources, such as service failures, attackers, or network connections. The transitions of states capture the connections of different cyber risks, and they are affected by the user's actions of protections at different times. 

To further mitigate the cyber risks, the user has a choice of purchasing cyber-insurance. After paying a premium, the user could receive financial coverages from the insurer to reimburse his losses caused by various cyber risks as shown in Fig. \ref{fig:OverviewUserInsurer}. The objective of the user is to find an optimal deployment of protections and cyber-insurance that minimizes his cyber-losses. 

\begin{figure}[http]
\centering
\includegraphics[width=0.9\textwidth]{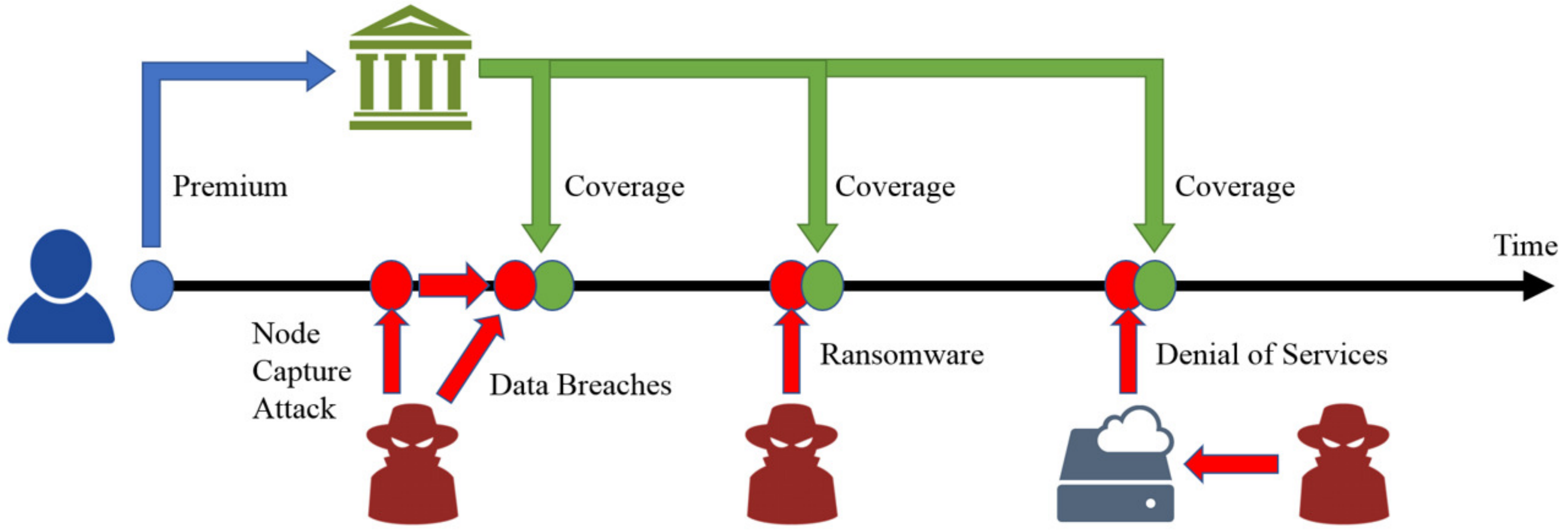}
\caption{Cyber-insurance example. The blue, red, and green icons represent user, attacker, and insurer, respectively. A user pays a premium to an insurer to purchase the cyber-insurance. Then, the user could receive financial coverages from the insurer to cover part of his losses caused by cyberattacks.}
\label{fig:OverviewUserInsurer}
\end{figure}

A rational user selects a cyber-insurance from which he could benefit more, i.e., contracts with a low premium and a high coverage. However, an insurer tends to offer an insurance contract that has a high premium and a low coverage, as the insurer aims to maximize his operating profit. Moreover, similar to the traditional insurance scenarios, the insurer is not aware of the local protections of the user, and an inappropriate insurance contract could largely damage the insurer's profit. 

We address such conflicting interests and the information asymmetry between the user and the insurer with a moral hazard type of principal-agent problem  \cite{holmstrom1979moral,holmstrom1982moral,grossman1983analysis,laffont2009theory}. The analysis, as well as the solution of the problem, is important to study the impacts of cyber-insurance to the user and design effective insurance contracts. The major contributions of this work are summarized as follows:

\begin{itemize}
\item We integrate Markov Decision Processes (MDP) into a moral hazard type of principal-agent model to  investigate the impacts of cyber-insurance on the user's cyber risks and design effective cyber-insurance contracts for the insurer. 
\item We fully characterize a case between a two-state two-action user and a linear coverage insurer. The results of this case indicate that the optimal insurance contracts follow linear insurance contract principle and zero-operating profit principle. The analysis also demonstrates the existence of risk compensation and Peltzman effect in cyber-insurance. 
\item We develop computational tools to solve problems involving multiple cyber-risk states, various protection choices, and complex insurance contracts. Our numerical experiments illustrate risk compensation, Peltzman effect, and zero-operating profit principle in cases of a four-state three-action user under linear coverage insurance and threshold coverage insurance.
\end{itemize}

\subsection{Organization of the Paper}
The rest of this paper is organized as follows. Section \ref{sec:RelatedWorks} presents the related works. Section \ref{sec:User} and Section \ref{sec:Insurer} discuss the user's problem and the insurer's problem, respectively. Section \ref{sec:Case} presents a case study of a linear coverage insurance contract on a two-state two-action user. Section \ref{sec:Num} and Section \ref{sec:Con} present numerical results and concluding remarks, respectively. Appendices A, B, and C provide the proofs of the Proposition \ref{pro:UserH}, and Theorem \ref{the:UserUniqueness}, and Proposition \ref{pro:UserSwitchingCoverageLevel}, respectively. We provide a summary of notations in the following table for convenience.

\begin{table}[http]
\small
\renewcommand{\arraystretch}{1}
\begin{center}
\begin{tabular}{cc}
\hline
\multicolumn{2}{c}{Summary of Notations}    
\\ \hline  \multicolumn{1}{c|}{$t$} & Time $t$
\\ \multicolumn{1}{c|}{$\mathcal{S}$, $N$} & Set and Number of All Possible States
\\ \multicolumn{1}{c|}{$S_n$} & State $n$ ($1\leq n \leq N$)
\\ \multicolumn{1}{c|}{$s$, $s_t$} & State, State at Time $t$ ($s, s_t\in\mathcal{S}$)
\\ \multicolumn{1}{c|}{$\mathcal{X}$} & Set of Direct Losses at All Possible States
\\ \multicolumn{1}{c|}{$X_n$} & Direct Loss at State $S_n$
\\ \multicolumn{1}{c|}{$x$, $x_t$} & Direct Loss, Direct Loss at Time $t$
\\ \multicolumn{1}{c|}{$p(s_t, s_{t+1})$} & Transition Probability from $s_t$ to $s_{t+1}$
\\ \multicolumn{1}{c|}{$\mathcal{A}$, $M$} & Set and Number of All Possible Protections
\\ \multicolumn{1}{c|}{$A_m$} & Protection $m$ ($1\leq m \leq M$)
\\ \multicolumn{1}{c|}{$a$, $a_t$} & Protection, Protection at Time $t$ ($a, a_t\in\mathcal{A}$) 
\\ \multicolumn{1}{c|}{$p(s_t, a_t, s_{t+1})$} & Transition Probability from $s_t$ to $s_{t+1}$ under $a_t$
\\ \multicolumn{1}{c|}{$c(a)$} & Cost Function
\\ \multicolumn{1}{c|}{$\alpha_s$} & Stationary State Protection at State $s$ ($\alpha_s \in \mathcal{A}$)
\\ \multicolumn{1}{c|}{$\Omega$} & Set of All Possible Stationary Protection Policies
\\ \multicolumn{1}{c|}{$\pi$} & Stationary Protection Policy ($\pi(s) = \alpha_s,\forall s\in\mathcal{S}$)
\\ \multicolumn{1}{c|}{$\rho$} & Value of Transition Probabilities
\\ \multicolumn{1}{c|}{$\mathcal{R}$} & Set of All Possible Coverage Functions
\\ \multicolumn{1}{c|}{$r(x)$, $K$} & Coverage Function ($r\in\mathcal{R}$), Premium ($K\in\mathbb{R}_{\geq 0}$)
\\ \multicolumn{1}{c|}{$r_0(x)$} & Zero Coverage Function ($r_0(x)=0,\forall x\in \mathbb{R}_{\geq 0}$)
\\ \multicolumn{1}{c|}{$R$} & Coverage Level ($r(x)=Rx$)
\\ \multicolumn{1}{c|}{$l(s,a,r)$} & Effective Loss Function
\\ \multicolumn{1}{c|}{$V(s,\pi,r)$} & Expected Cumulative Effective Loss Function
\\ \multicolumn{1}{c|}{$\pi_r^*$} & Optimal Stationary Protection Policy Under Coverage $r$
\\ \multicolumn{1}{c|}{$\mathcal{S}_{GB}$} & Set of Two States ($\mathcal{S}_{GB}=\{S_G, S_B\}$)
\\ \multicolumn{1}{c|}{$S_G$, $S_B$} & Good State, Bad State
\\ \multicolumn{1}{c|}{$s$, $s^c$} & One State, The Other State ($s, s^c\in\mathcal{S}_{GB}; s^c \neq s$)
\\ \multicolumn{1}{c|}{$X_G$, $X_B$} & Direct Losses at Good State, Bad State
\\ \multicolumn{1}{c|}{$\mathcal{A}_{HL}$} & Set of Two Actions ($\mathcal{A}_{HL}=\{A_H, A_L\}$)
\\ \multicolumn{1}{c|}{$A_H$, $A_L$} & Strong Protection, Weak Protection
\\ \multicolumn{1}{c|}{$\alpha_G$, $\alpha_B$} & Stationary State Protections at Good State, Bad State
\\ \multicolumn{1}{c|}{$\alpha_s$, $\alpha_{s^c}$} & Stationary State Protections
\\ \multicolumn{1}{c|}{$R_G$, $R_B$} & Threshold Coverage Levels at Good State, Bad State
\\ \hline
\end{tabular}
\end{center}
\end{table}

\section{Related Works}
\label{sec:RelatedWorks}
Recently, with fast-growing types and amounts of the networked devices and shortages of effective and state-of-art defense methods, cyber-insurance has drawn huge attention as it can transfer the unexpected cyber risks to the insurance companies   \cite{bohme2005cyber,kesan2005cyberinsurance,bohme2006models,grossklags2008secure,bohme2010modeling,shetty2010competitive,pal2014will,tosh2017risk,kesan2017strengthening,zhang2017bi,laszka2018cyber}. The existing insurance framework could bring useful insights on modeling the cyber-insurance  \cite{ehrlich1972market,kesan2005cyberinsurance}. The moral hazard models in the economics literature are good tools to capture the information asymmetry between the insured and the insurers \cite{rothschild1978equilibrium,holmstrom1979moral,holmstrom1982moral}. 
Various frameworks and methodologies have been brought up to investigate cyber-insurance contracts and their impacts to cyber risks. 

Several works have studied cyber-insurance through market-based approaches by analyzing the supply and demand relations between insurers and insureds  \cite{pal2014will,bohme2005cyber,bohme2006models,bohme2010modeling}. In  \cite{pal2014will}, Pal et al., have analyzed regulated monopolistic and competitive cyber-insurance markets, and showed that cyber-insurance can improve the network security but the insurer can make zero expected profits in monopoly markets. In  \cite{bohme2005cyber,bohme2006models,bohme2010modeling}, B{\"o}hme et al., have presented several market models of cyber-insurance with the consideration of interdependency between cyber risks and information asymmetries between insurers and insureds, and showed analytical results on the impacts of cyber-insurance to cyber-security and the vialibity of a market for cyber-insurance.

Game theory has been used to capture the interactions between insurers and insureds of cyber-insurance \cite{laszka2018cyber,grossklags2008secure,zhang2017bi}. In  \cite{laszka2018cyber}, Laszka et al., have used a two-player signaling game to capture the information asymmetry between a potential client and an insurer, and further studied incentives for auditing potential clients before cyber-insurance premium calculations. In  \cite{grossklags2008secure},  Grossklags et al., have presented several security games to capture the decision-making of network users on protections and insurance. The equilibrium analysis shows that users may seek to self-protect themselves at just slightly above the lowest protection level in the weakest-target game. In  \cite{zhang2017bi}, Zhang et al. have studied the interactions between insureds, attackers, and insurers with a bi-level game-theoretic framework in a networked environment and demonstrated the impacts of network connections to the three types of players. 

Most previous works have focused on the information asymmetry and interdependencies of cyber risks, however, their models have not captured the dynamics and correlations of the cyber risks, which have been studied with different methodologies and models \cite{stallings1995network,perrig2004security,zhu2012guidex,poolsappasit2012dynamic,kim2004measurement}. In  \cite{poolsappasit2012dynamic}, Poolsappasit et al. have used a bayesian attack graphs model to analyze the network security risk assessment and mitigation. In  \cite{kim2004measurement}, the authors have used a differential epidemic model to capture the spreading of viruses and worms in computer networks. These works aim to reduce the impacts of cyber risks through local protections, such as firewalls  \cite{cheswick2003firewalls}, intrusion detection  \cite{jones2000computer}, or moving target defenses  \cite{jajodia2012moving}, which cannot fully mitigate the risks of cyberattacks. 

In this work, we focus on studying the dynamics and correlations of the cyber risks and analyzing the impacts of the cyber-insurance to both the insureds and the insurers. We first capture the cyber risks as well as the user's deployments of local protections with Markov decision processes, which have been used variously to analyze cybersecurity  \cite{wu2010optimal,shen2007adaptive}. We then use the existing moral hazard type of principal-agent model to capture the interactions between the user and the insurer with incomplete information. The analysis of both the optimal insurance contract and the user's response to it provides useful insights on the designation of the cyber-insurance contracts in the real world.

\section{User's Optimal Protection Policies}
\label{sec:User}
We use discrete Markov decision processes (MDP) to capture the evolvements of the user's cyber risks with time, and an illustration is provided in Fig. \ref{fig:OverviewUserInsurerMDP}. Let $s_t \in \mathcal{S}$ denote the user's cyber-risk state at time $t\in\mathbb{Z}_{\geq 0}$, where $\mathcal{S} \equiv\{S_n | 1\leq n \leq N\}$ is the set of all possible cyber-risk states. Different cyber-risk states may incur various types of losses, e.g., data breaches, physical device damages, and compromised financial accounts. In this paper, we consider that all types of losses are measurable and can be quantified by monetary direct losses. We assume that each cyber-risk state $S_n\in\mathcal{S}$ is associated with a fixed direct loss $X_n\in\mathbb{R}_{\geq 0}$, and the user's direct loss at time $t$ can be denoted by $x_t\in\mathcal{X}$, where $\mathcal{X}\equiv\{X_n| 1\leq n \leq N\}$. 

\begin{figure}[http]
\centering
\includegraphics[width=0.9\textwidth]{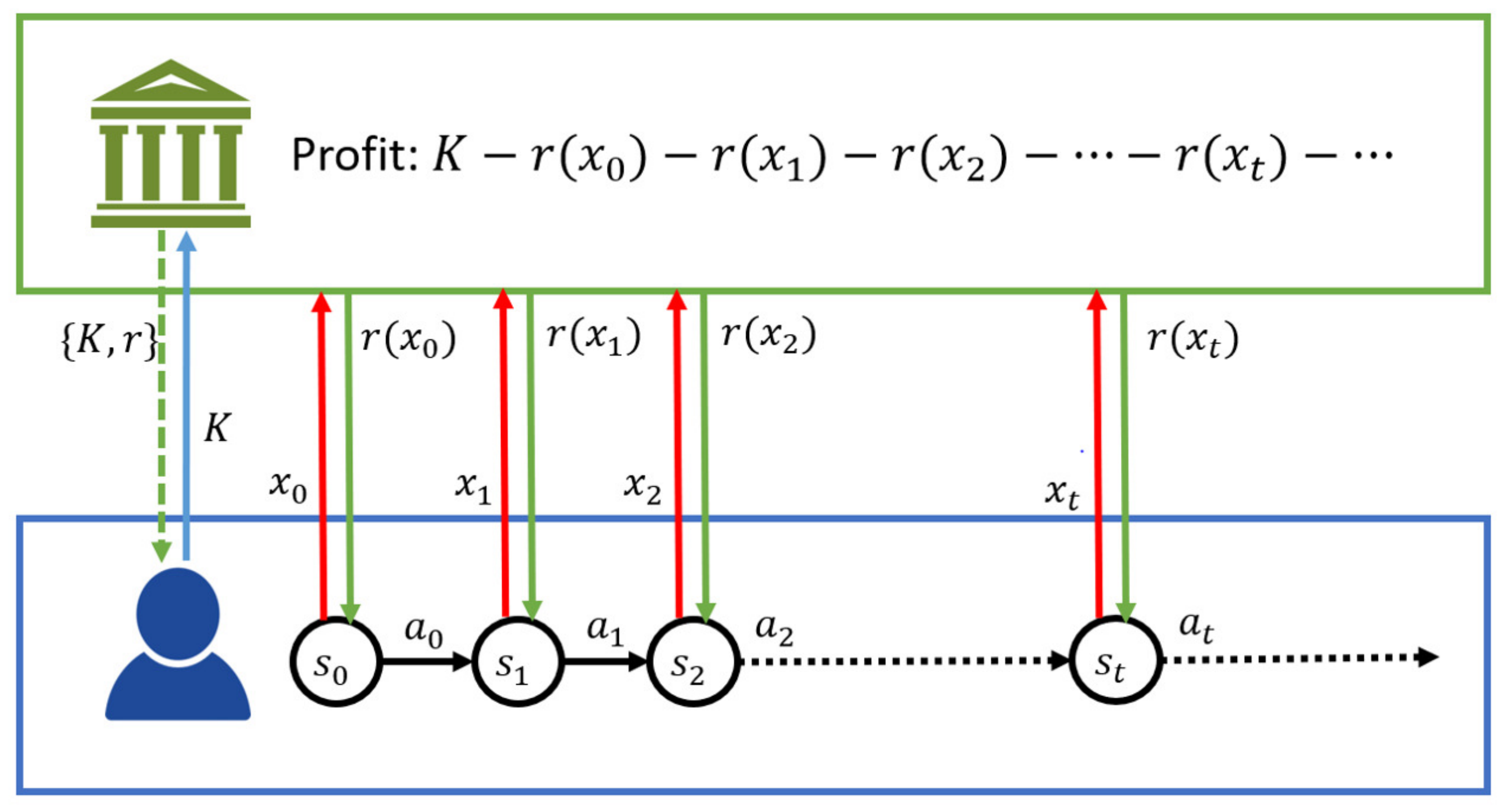}
\caption{Illustration of cyber-insurance. The dynamics of the user's cyber risks are captured by MDP with $s_t$ denoting the cyber-risk state at time $t$, which is associated with a direct loss $x_t$. The user can choose various protections $a_t$ to reduce the future losses. The objective of the user is to find the optimal protection sequence $\{a_t\}_{t\geq 0}$ which minimizes his cumulative losses. The user can also purchase cyber-insurance to mitigate his losses. The insurer first announces the insurance contract $\{K,r\}$, where $K$ and $r$ indicate the premium and the coverage function, respectively. The user can decide whether to purchase the insurance or not. If the user chooses to purchase the insurance, he must pay a premium $K$, and when he faces a loss of $x$, the insurer should provide a coverage of $r(x)$ to him. The objective of the insurer is to maximize his profit. Note that the insurer has no information of the user's protection sequences. }
\label{fig:OverviewUserInsurerMDP}
\end{figure}

The user can adopt different protections, such as firewalls, intrusion detection systems, and moving-target defenses, to reduce the possibilities of entering cyber-risk states that can incur severe losses. Let $a_t\in\mathcal{A}$ denote the protections at time $t$, where $\mathcal{A}\equiv \{A_m| 1\leq m \leq M\}$ is the set of all available protections. The transition probability $p(s_t, a_t, s_{t+1})$ denotes the probability that the user goes to state $s_{t+1}$ at time $t+1$ when he is currently in state $s_t$ and adopts protection $a_t$, which naturally captures the correlations among different cyber-risk states under different protections. Note that $\sum_{n=1}^N p(s_t, a_t, S_n) = 1$ as the user can only enter states within $\mathcal{S}$ at time $t+1$. 

We further provide two examples to illustrate the states $\mathcal{S}$ and protections $\mathcal{A}$ of the user.   
\begin{example}
Suppose a customer whose computer faces threats of Ransomware. In this example, the customer has $\mathcal{S} = \{S_1, S_2\}$ and $\mathcal{A} = \{A_1, A_2\}$. States $S_1$ and $S_2$ denote that the computer is secure and compromised, respectively. The customer can choose to do nothing $A_1$ or add firewalls $A_2$. The computer has a lower probability of facing Ransomware, i.e., entering state $S_2$, if the customer deploys firewalls. When the computer is compromised, the customer needs to either pay the money or replace the computer, which can be covered if he has purchased cyber-insurance. 
\end{example}

\begin{example}
Consider a cloud center who aims to protect itself from the damages caused by potential attackers. In this example, the cloud center has $\mathcal{S} = \{S_1, S_2, S_3\}$ and $\mathcal{A} = \{A_1, A_2, A_3, A_4\}$. State $S_1$ denotes the situation when it is safe and faces no cyberattacks. However, the cloud center may encounter data breaches and denial-of-services, which are represented by states $S_{2}$ and $S_{3}$, respectively. Each state $S_n\in\mathcal{S}$ is associated with a direct loss $X_{n}$. For example, at time $t$, $s_t = S_{2}$ indicates that the cloud center faces data breaches which inflict $X_{2}$ direct losses to it. Specially, the direct loss $X_1 = 0$ at state $S_1$, which indicates that the cloud center has no loss when it faces no cyberattacks. To defend against these cyberattacks, the cloud center may deploy firewalls, intrusion detection systems, and moving-target defense, which are represented by protections $A_{2}$, $A_{3}$, and $A_{4}$, respectively. Specially, the cloud center can also choose to do nothing, which is denoted as $A_{1}$. The cloud center has smaller probabilities of entering states with high losses if he deploys protections, however, these protections are also costly. The cloud center can also purchase cyber-insurance to cover part of its losses and help it recover from cyber-incidents. The objective of the cloud center is to find an optimal deployment of protections and cyber-insurance such that its future cumulative losses are minimized. Cases involve other cyber risks or protections can be extended through increasing the size of $\mathcal{S}$ and $\mathcal{A}$. 
\end{example}

Besides the protections, the user can also mitigate his losses through purchasing cyber-insurance. After paying a premium to an insurer, the user could receive a coverage of $r(x_t)$ from the insurer when he faces a direct loss of $x_t$, where $r:\mathbb{R}_{\geq 0}\rightarrow \mathbb{R}_{\geq 0}$ is the coverage function of the insurance. The objective of the user is to find an optimal sequence of protections $\{a_t\}_{t\in\mathbb{Z}_{\geq 0}}$ that minimizes the expected cumulative effective losses given the initial state $s_0\in\mathcal{S}$, which can be captured as 
\begin{equation}
\label{eq:UserTotalEffectiveLosses}
\min\limits_{\{a_t\}} \mathbb{E}\left\lbrace \sum\limits_{t=0}^{\infty} \delta^t  \left( x_t - r(x_t) + c(a_t) \right) \Big| s_0   \right\rbrace,
\end{equation}
where function $c(a_t)$ returns the cost of protection $a_t$, and $\delta \in (0,1)$ is the discount factor which indicates that future losses are valued less at time $0$. 

In this paper, we consider that the user decides his protections contingent on his current states. Such feedback strategy allows the user to maintain his security level by adopting the necessary protections that can reduce the losses from cyberattacks and save the costs of protections at the same time. The strategy is usually denoted by a stationary protection policy $\pi:\mathcal{S}\rightarrow \mathcal{A}$, e.g., $\pi(S_n)=A_m$ indicates that the user always takes protection $A_m$ at state $S_n$. As a result, the expected cumulative effective losses of the user under a stationary protection policy can be captured as
\begin{equation}
\label{eq:UserObjJ}
\begin{array}{l}
V(s_0,\pi,r) =\mathbb{E}\left\lbrace \sum\limits_{t=0}^{\infty} \delta^t  \left( x_t - r\left(x_t\right) + c\left(\pi\left(s_t\right)\right) \right) \Big| s_0 \right\rbrace.
\end{array}
\end{equation}

The user aims to find an optimal stationary protection policy $\pi_r^*\in\Omega$ that minimizes his expected cumulative effective losses given the coverage function $r$, and such objective can be captured as
\begin{equation}
\label{eq:UserObj}
\begin{array}{l}
\pi_r^* \in \arg \min\limits_{\pi\in\Omega}V(s_0,\pi,r), 
\end{array}
\end{equation}
where $\Omega$ denotes the set of all possible stationary policies. 

A rational user purchases the insurance only when the expected cumulative effective losses plus the premium under the insurance is lower than the losses without insurance, which can be captured as
\begin{equation}
\label{eq:UserIR}
 V(s_0,\pi_r^*,r) + K \leq V(s_0,\pi_{r_0}^*,r_0),   
\end{equation}
where $K\in\mathbb{R}_{\geq 0}$ is the premium of the insurance and $r_0$ indicates a zero coverage function, i.e., $r_0(X) = 0$ for all $X\in\mathbb{R}_{\geq 0}$, which corresponds to the case when there is no insurance. The fact that the user purchases the insurance only when inequality (\ref{eq:UserIR}) is satisfied must be considered by the insurer while designing effective insurance contracts. 

The optimal protection policy $\pi_r^*$ could be obtained by solving (\ref{eq:UserObj}) with either dynamic programming or linear programming \cite{filar2012competitive}, and we summarize both approaches in the following subsections. 
\subsection{Dynamic Programming Approach}
Recall equation (\ref{eq:UserObjJ}), let us define the loss function $l(s_t,a_t,r)=  x_t - r\left(x_t\right) + c\left(a_t\right) $ which indicates the effective loss at time $t$ under the coverage function $r$. Note that $l(s_t,a_t,r)$ does not take $x_t$ as variable since the direct loss $x_t$ is uniquely determined by the user's cyber-risk state $s_t$. Thus, we can express the expected cumulative effective losses as
\begin{equation}
\label{eq:UserValue}
\begin{array}{l}
V(s_0,\pi,r) = \mathbb{E}\left\lbrace  \sum\limits_{t=0}^\infty \delta^t l(s_t,\pi(s_t),r) \left| s_0 \right. \right\rbrace   =  l(s_0,\pi(s_0),r)   + \delta \sum\limits_{s'\in\mathcal{S}}  p(s_0,\pi(s_0),s') V(s',\pi,r),
\end{array}
\end{equation}
where $l(s_0,\pi(s_0),r)$ and $\delta \sum\limits_{s\in\mathcal{S}}  p(s_0,\pi(s_0),s) V(s,\pi,r)$ capture the effective loss at time $0$ and the future expected cumulative effective losses, respectively. As a result, given a coverage function $r$, the optimal protection policy $\pi_r^*$ can be found by the following dynamic programming operators  \cite{filar2012competitive}.
\begin{equation}
\label{eq:UserDynamicProgrammingPi}
\begin{array}{l}
\pi_r^*(s) \in \arg\min\limits_{a\in\mathcal{A}}  \left\lbrace      \begin{array}{l}
l(s,a,r)   +  \delta \sum\limits_{s'\in\mathcal{S}}  p(s,a,s') V(s',\pi_r^*, r)  
\end{array}           \right\rbrace,
\end{array}
\end{equation}
\begin{equation}
\label{eq:UserDynamicProgrammingV}
V(s,\pi_r^*, r)   = l(s,\pi_r^*(s),r)   + \delta \sum\limits_{s'\in\mathcal{S}}  p(s,\pi_r^*(s),s') V(s',\pi_r^*, r).
\end{equation}
By iterating (\ref{eq:UserDynamicProgrammingPi}) and (\ref{eq:UserDynamicProgrammingV}) for all states $s\in\mathcal{S}$ until no further changes take place, we can achieve $\pi_r^*$ and $V(s,\pi_r^*, r)$, and the convergence to the optimum is guaranteed  \cite{filar2012competitive}.

\subsection{Linear Programming Approach}
Besides the dynamic programming, we can also use linear programming to solve the user's problem (\ref{eq:UserObj}). Problem (\ref{eq:UserObj}) can be reformulated into a linear programming problem in the standard form as \cite{filar2012competitive} 
\[\begin{array}{c}
\min\limits_{\eta} \mathbf{d}^T \eta  \\
\begin{array}{cc}
{\text{s.t.}}&{O\eta = \mathbf{b},\eta \geq \mathbf{0},}
\end{array}
\end{array}\]
with its dual problem
\[\begin{array}{c}
\max\limits_{ \theta } \mathbf{b}^T\theta \\
\begin{array}{cc}
{\text{s.t.}}&{ \mathbf{d}-  O^T\theta \geq \mathbf{0}, }
\end{array}
\end{array}\]
where $\eta \in\mathbb{R}^{NM\times 1}$ and $\theta \in\mathbb{R}^{N\times 1}$ are denoted as the prime variable and the dual variable, respectively. Note that $N$ and $M$ are the sizes of $\mathcal{S}$ and $\mathcal{A}$, respectively. Vector $\mathbf{b}\in\mathbb{R}^{N\times 1}$ is a column vector of size $N$ with all the elements equal to $1$. Vector $\mathbf{d}\in\mathbb{R}^{NM\times 1}$ is a column vector of size $NM$ which captures the per-state and per-action losses, and the $(N(n-1) + m)$-th element of it equals $l(S_n,A_m,r)$, where $1\leq n \leq N$ and $1\leq m \leq M$. Matrix $O = E-\delta P$, where matrix $E\in\mathbb{R}^{N\times NM}$ has that $E_{n,N(n-1)+m}=1$ for $1\leq n \leq N$ and $1\leq m \leq M$ and all the other elements are $0$, and matrix $P\in\mathbb{R}^{N\times NM}$ is the transition probability matrix where $P_{n',N(n-1)+m} = p(n,a_m,n')$, $1\leq n \leq N$, $1\leq n' \leq N$, and $1\leq m \leq M$.

The optimal primal variable $\eta^*$ represents the optimal state-action frequencies; the optimal dual variable $\theta^*$ represents the expected cost-to-go values of the states for the given coverage function $r$, i.e., $\theta_n^* = V(S_n,\pi_r^*, r)$ for $1\leq n \leq N$. After solving the dual problem, we can find the optimal protection policy $\pi_r^*$ by plugging  $V(S_n,\pi_r^*, r)$ into (\ref{eq:UserDynamicProgrammingPi}).

\section{Insurer's Optimal Insurance Contracts}
\label{sec:Insurer}  
In this section, we present and analyze the insurer's problem of designing cyber-insurance contracts. An illustration of the interactions between the user and the insurer has been provided in Fig. \ref{fig:OverviewUserInsurerMDP}. Note that the insurer first announces the insurance contract $\{K,r\}$, and the user then makes the decision of purchasing the insurance based on the expected cumulative effective losses under that insurance contract. If the user chooses to purchase the insurance, the insurer instantly earns a profit of $K$ at time $0$, but the insurer is required to pay the coverage of $r(x_t)$ when the user faces a loss of $x_t$ at time $t$. As a result, the insurer's operating profit can be captured as $K - \mathbb{E}\left\{\sum_{t= 0}^\infty \delta^t r(x_t)\big| s_0\right\}$, where $\mathbb{E}\{\sum_{t= 0}^\infty \delta^t r(x_t)\big| s_0\} $ denotes the expected cumulative coverage provided by the insurer to the user. The objective of the insurer is to find an optimal insurance contract $\{K^*,r^*\}$ that maximizes his operating profit. As a result, the insurer's problem can be captured as
\begin{equation}
\label{eq:InsurerObj}
\begin{array}{l}
\max\limits_{\{K,r\}} K - \mathbb{E}\left\{\sum\limits_{t= 0}^\infty \delta^t r(x_t)\big| s_0\right\} \\
\begin{array}{lll}
{\text{s.t.} } & {K - \mathbb{E}\left\{\sum\limits_{t= 0}^\infty \delta^t r(x_t)\big| s_0\right\} \geq 0;} & {(\ref{eq:InsurerObj}a)}\\
{} & {V(s_0,\pi_r^*, r) + K \leq V(s_0,\pi_{r_0}^*, r_0). }  & {(\ref{eq:InsurerObj}b)}
\end{array}
\end{array}
\end{equation} 
Constraint (\ref{eq:InsurerObj}a) captures the insurer's individual rationality that he chooses not to provide the insurance if he has a negative profit. Constraint (\ref{eq:InsurerObj}b) captures the user's individual rationality on purchasing the insurance and it comes from inequality (\ref{eq:UserIR}).

By solving problem (\ref{eq:InsurerObj}), the insurer can find an optimal insurance contract which maximizes his operating profit and is acceptable by the user. After combing the user's problem and the insurer's problem, the interactions of the user and the insurer can be captured by the following principal-agent problem.
\begin{equation}
\label{eq:InsurerUserTogether}
\begin{array}{c}
\max\limits_{\{K,r\}} K - \mathbb{E}\left\{ \sum\limits_{t = 0}^\infty \delta^t r(x_t)\big| s_0\right\} \\
\begin{array}{lll}
{\text{s.t.} }&{ K - \mathbb{E}\left\{ \sum\limits_{t = 0}^\infty \delta^t r(x_t)\big| s_0\right\} \geq 0;}&{(\ref{eq:InsurerUserTogether}a)} 
\\ {}&{ V(s_0,\pi_r^*, r) + K \leq V(s_0,\pi_{r_0}^*, r_0);}&{(\ref{eq:InsurerUserTogether}b)} 
\\  {}&{
\pi_r^*  \in  \arg\min\limits_{\pi\in\Omega}V(s_0,\pi,r).}&{(\ref{eq:InsurerUserTogether}c)}  
\end{array}  \end{array}
\end{equation}

Problem (\ref{eq:InsurerUserTogether}) is an optimization problem nested with various sub-optimization problems. The solution of Problem (\ref{eq:InsurerUserTogether}) captures both the user's objective of minimizing his expected cumulative effective losses and the insurer's objective of maximizing his own profit with the consideration of the user's rational choice of purchasing the insurance. To find the solution of problem (\ref{eq:InsurerUserTogether}), we can first solve the user's problem (\ref{eq:UserObj}) and obtain the optimal protection policies $\pi_r^*$ and the corresponding losses $V(s_0,\pi_r^*,r)$ to the coverage function $r$, and then achieve $\{K^*, r^*\}$ by solving the insurer's problem (\ref{eq:InsurerObj}). 

We can simplify the insurer's problem (\ref{eq:InsurerObj}) by exploring the expected cumulative effective losses and the optimal protection policies as discussed in the following subsection.   
\subsection{Insurer's Problem: Simplifications and Direct Conclusions}
We first notice that the expected cumulative coverage is equal to the expected cumulative direct losses minus the expected cumulative effective losses, i.e., 
\[ \begin{array}{l}
\mathbb{E}\left\lbrace  \sum\limits_{t=0}^\infty \delta ^t r(x_t) \left| s_0 \right. \right\rbrace \\  = \mathbb{E}\left\lbrace  \sum\limits_{t=0}^\infty \delta ^t \left(x_t + c(\pi_r^*(s_t)) \right) \left| s_0  \right. \right\rbrace  - \mathbb{E}\left\lbrace  \sum\limits_{t=0}^\infty \delta ^t \left( x_t - r(x_t) +c(\pi_r^*(s_t)) \right)  \left| s_0\right. \right\rbrace
\\ = V(s_0,\pi_r^*,r_0) - V(s_0,\pi_r^*,r),
\end{array}   \]
where $V(s_0,\pi_r^*,r_0)$ can be interpreted as the expected cumulative effective losses given the optimal protection policy $\pi_r^*$ and the zero coverage function $r_0$. Thus, Problem (\ref{eq:InsurerObj}) can be rewritten as follows.
\begin{equation}
\label{eq:InSim}
\begin{array}{c}
\max\limits_{\{K,r\}} K - \left( V(s_0,\pi_r^*,r_0) - V(s_0,\pi_r^*,r)   \right) \\
\begin{array}{lll}
{\text{s.t. } }&{K - \left( V(s_0,\pi_r^*,r_0) - V(s_0,\pi_r^*,r) \right) \geq 0;}&{(\ref{eq:InSim}a)}
\\ {}&{V(s_0,\pi_r^*,r) + K \leq V(s_0,\pi_{r_0}^*,r_0).}&{(\ref{eq:InSim}b)}
\end{array}
\end{array} 
\end{equation}
Constraint (\ref{eq:InSim}b) indicates that the maximum premium that can be charged by the insurer for a coverage function $r$ is
\begin{equation}
\label{eq:InsurerTmax}
K_{\max} = V(s_0,\pi_{r_0}^*,r_0) -  V(s_0,\pi_r^*,r).
\end{equation}
The user chooses not to purchase the insurance with a premium $K>K_{\max}$ because the losses and the premium under the insurance are higher than the losses without insurance. 

As a result, problem (\ref{eq:InsurerObj}) is equivalent to the following problem after letting $K$ be equal to $K_{\max}$ and plugging (\ref{eq:InsurerTmax}) into its objective function and constraint.
\begin{equation}
\label{eq:InsurerObjRe}
\begin{array}{l}
\max\limits_{r\in\mathcal{R}} V(s_0,\pi_{r_0}^*,r_0) - V(s_0,\pi_r^*,r_0)\\
\begin{array}{ll}
{\text{s.t. } }&{V(s_0,\pi_{r_0}^*,r_0) - V(s_0,\pi_r^*,r_0) \geq 0,}
\end{array}
\end{array}  
\end{equation}
where $\mathcal{R}$ denotes the set of all possible coverage functions, and the constraint indicates that the profit of the insurer cannot be negative. After solving (\ref{eq:InsurerObjRe}), we can find the optimal coverage function $r^*$, and then the optimal premium can be computed through (\ref{eq:InsurerTmax}). Similarly, the principal-agent problem (\ref{eq:InsurerUserTogether}) can also be rewritten as 
\begin{equation}
\label{eq:InsurerUserTogetherRe}
\begin{array}{c}
\max\limits_{r \in\mathcal{R}} V(s_0,\pi_{r_0}^*,r_0) - V(s_0,\pi_r^*,r_0) \\
\begin{array}{lll}
{\text{s.t.} }&{V(s_0,\pi_{r_0}^*,r_0) -  V(s_0,\pi_r^*,r_0) \geq 0;}&{(\ref{eq:InsurerUserTogetherRe}a)} \\  {}&{
\pi_r^*  \in  \arg\min\limits_{\pi\in\Omega}V(s_0,\pi,r).}&{(\ref{eq:InsurerUserTogetherRe}b)}  
\end{array}  \end{array}
\end{equation}
Comparing to (\ref{eq:InsurerObj}) and (\ref{eq:InsurerUserTogether}), we only need to find the optimal protection policies $\pi_r^*$ to obtain the optimal insurance contract $\{K^*, r^*\}$ through (\ref{eq:InsurerObjRe}) and (\ref{eq:InsurerUserTogetherRe}). 

One useful insight regarding the operating profit could be obtained without solving (\ref{eq:InsurerObjRe}) or (\ref{eq:InsurerUserTogetherRe}), which is summarized in the following remark and proposition. 
\begin{remark}
\label{rem:InsurerZeroOperatingProfitWithoutSolving}
Any coverage function $r$ that yields $\pi_r^* = \pi_{r_0}^*$, i.e., the user has the same optimal protection policy between the case under the coverage function $r$ and the case under no insurance, is a feasible solution with the corresponding premium $K = V(s_0,\pi_{r_0}^*,r_0) -  V(s_0,\pi_{r}^*,r) = V(s_0,\pi_{r_0}^*,r_0) -  V(s_0,\pi_{r_0}^*,r) \geq 0$ as $V(s_0,\pi_{r_0}^*,r) \leq V(s_0,\pi_{r_0}^*,r_0)$, and the insurer has a zero operating profit under that insurance contract as $V(s_0,\pi_{r_0}^*,r_0) -  V(s_0,\pi_r^*,r_0)= V(s_0,\pi_{r_0}^*,r_0) -  V(s_0,\pi_{r_0}^*,r_0) = 0 $. 
\end{remark}
\begin{proposition}
\label{pro:InsurerOptimalGeneral}
Any insurance contract $\{K,r\}$ that yields $\pi_r^* = \pi_{r_0}^*$ and meets (\ref{eq:InsurerTmax}) is optimal for the insurer, and the insurer has a zero operating profit under that contract. 
\end{proposition}
\begin{proof}
The operating profit of the insurer has that $V(s,\pi_{r_0}^*,r_0) - V(s,\pi_{r}^*,r_0) \leq 0 $ from (\ref{eq:InsurerUserTogetherRe}b), i.e., $ \pi_{r_0}^* \in \arg \min \limits_{\pi \in \Omega} V(s,\pi,r_0)$. Thus, the maximum profit that the insurer can achieve is $0$. As a result, if the user has $\pi_r^* = \pi_{r_0}^*$ under an insurance contract $\{K, r\}$, that contract is feasible from Remark \ref{rem:InsurerZeroOperatingProfitWithoutSolving} and it is also optimal.
\end{proof}
Remark \ref{rem:InsurerZeroOperatingProfitWithoutSolving} indicates that the insurer has a zero operating profit when the user has the same protection policies with or without insurance, which explains market neutrality. Proposition \ref{pro:InsurerOptimalGeneral} indicates that the insurance contract in Remark \ref{rem:InsurerZeroOperatingProfitWithoutSolving} is optimal for the insurer. We denote this conclusion as the zero operating profit principle. 

\section{Case Study: Two-State Two-Action User and Linear Coverage Insurer}
\label{sec:Case}
In this section, we present a representative case where the user has two states and two actions and the insurer provides the linear coverage. Analysis of this case provides structural insights of the insurance contracts. Recall Section \ref{sec:User}, the user in this case has the set of states $\mathcal{S}_{GB}\equiv \{S_G,S_B\}$, where $S_G$ and $S_B$ indicate good state and bad state, respectively. The losses that associated with the states can be further identified as $X_G$ and $X_B$. The difference between the good state and the bad state is that the user has lower losses at the good state than that at the bad state, i.e., $0 \leq X_G < X_B$. 

To reduce the losses, the user can choose to take a strong protection $A_H$ or a weak protection $A_L$, in other words, the user has the action set $\mathcal{A}_{HL}=\{A_H,A_L\}$. We further use shorthand notations $C_H$ and $C_L$ to represent the costs of protections $A_H$ and $A_L$, respectively, i.e, $c(A_H)= C_H$ and $c(A_L) = C_L$. The differences between a strong protection and a weak protection can be identified in detail as follows:
\begin{itemize}
\item $p(s,A_H,S_B)< p(s,A_L,S_B),\forall s\in\mathcal{S}_{GB}$, which indicates that the user has a higher probability of going to the bad state when he has a weak protection. 
\item $p(s,A_L,S_G)< p(s,A_H,S_G), \forall s\in\mathcal{S}_{GB}$, which indicates that the user has a higher probability of going to the good state when he has a strong protection. 
\item $0 \leq C_L < C_H$, which indicates that the cost of a strong protection is higher than the cost of a weak protection. 
\end{itemize}
These differences capture the fact that a strong protection can make the user more secure but its cost is also higher. 

With two states and two actions, the user has only four possible stationary protection policies, i.e., $\Omega = \{\Pi_{HH},\Pi_{HL},\Pi_{LH},\Pi_{LL}\}$, where 
\begin{itemize}
\item $\Pi_{HH}(S_G) = A_H$ and $\Pi_{HH}(S_B) = A_H$;
\item $\Pi_{HL}(S_G) = A_H$ and $\Pi_{HL}(S_B) = A_L$;
\item $\Pi_{LH}(S_G) = A_L$ and $\Pi_{LH}(S_B) = A_H$;
\item $\Pi_{LL}(S_G) = A_L$ and $\Pi_{LL}(S_B) = A_L$. 
\end{itemize}
An optimal protection policy $\pi^*\in\Omega$ can be achieved by solving Problem (\ref{eq:UserObj}) which minimizes the user's expected cumulative effective losses.

Besides protections, the user can also purchase the insurance to further mitigate his losses. We consider that the insurer offers a linear coverage with $R\in [0,1]$ denoting the coverage level of the insurance, i.e., $r(x) = Rx$. Specially, $R=0$ and $R=1$ indicate no coverage and full coverage, respectively. 

Methods in Sections \ref{sec:User} and \ref{sec:Insurer} can be used to find the optimal protection policy of the user and the optimal insurance contract for the insurer. Since there are only two states and two actions for the user, we can find them analytically. 

\subsection{User's Optimal Protection Policy}
We first introduce several notations to simplify representations. Since the user has only two states $S_G$ and $S_B$, we use $s^c  \neq s$ to denote the other state for a given state $s\in\mathcal{S}_{GB}$. Since the user adopts a stationary protection policy, i.e., he has fixed protections at each state, we identify his state protections as $\alpha_{G}$ and $\alpha_{B}$ for the good state and the bad state, respectively. We further define the action dependent expected cumulative effective loss function as follows
\begin{equation}
\label{eq:UserJBar}
\begin{array}{l}
\overline{V}(s,\alpha_s;\alpha_{s^c},R) \\ = l(s,\alpha_s,R)    + \delta p(s,\alpha_s,S_G) \overline{V}(S_G,\alpha_G;\alpha_B,R)  + \delta  p(s,\alpha_s,S_B) \overline{V}(S_B,\alpha_B;\alpha_G,R).
\end{array}
\end{equation}
\begin{remark}
\label{rem:UserActionDependentExpectedTotalEffectiveLosses}
For a protection policy $\pi$ that has $\pi(S_G)=\alpha_{G}$ and $\pi(S_B)=\alpha_{B}$, the expected cumulative effective loss function (\ref{eq:UserValue}) is equivalent to the action dependent expected cumulative effective loss function (\ref{eq:UserJBar}), i.e., 
\[V(S_G,\pi,R) = \overline{V}(S_G,\pi(S_G);\pi(S_B),R) = \overline{V}(S_G,\alpha_{G};\alpha_{B},R);  \]
\[V(S_B,\pi,R) = \overline{V}(S_B,\pi(S_B);\pi(S_G),R) = \overline{V}(S_B,\alpha_{B};\alpha_{G},R).  \]
\end{remark}
As a result, the dynamic programming operators (\ref{eq:UserDynamicProgrammingPi}) and (\ref{eq:UserDynamicProgrammingV}) can be written as
\begin{equation}
\label{eq:UserDynamicProgrammingPiReG}
\begin{array}{l}
\pi_R^*(S_G) \in \arg\min\limits_{\alpha_{G}\in\mathcal{A}_{HL}} \overline{V}(S_G,\alpha_{G};\pi_R^*(S_B),R);
\end{array}
\end{equation}
\begin{equation}
\label{eq:UserDynamicProgrammingPiReB}
\begin{array}{l}
\pi_R^*(S_B) \in \arg\min\limits_{\alpha_{B}\in\mathcal{A}_{HL}} \overline{V}(S_B,\alpha_{B};\pi_R^*(S_G),R) ,
\end{array}
\end{equation}
where
\begin{equation}
\label{eq:UserDynamicProgrammingVReG}
\begin{array}{l}
\overline{V}(S_G,\alpha_G;\alpha_B,R)  = l(S_G,\alpha_G,R)  \\  \ \ \ \ \ \ \ \  + \delta p(S_G,\alpha_G,S_G) \overline{V}(S_G,\alpha_G;\alpha_B,R)  + \delta  p(S_G,\alpha_G,S_B) \overline{V}(S_B,\alpha_B;\alpha_G,R);
\end{array}
\end{equation}
\begin{equation}
\label{eq:UserDynamicProgrammingVReB}
\begin{array}{l}
\overline{V}(S_B,\alpha_B;\alpha_G,R)  = l(S_B,\alpha_B,R)  \\  \ \ \ \ \ \ \ \  + \delta p(S_B,\alpha_B,S_G) \overline{V}(S_G,\alpha_G;\alpha_B,R)  + \delta  p(S_B,\alpha_B,S_B) \overline{V}(S_B,\alpha_B;\alpha_G,R).
\end{array}
\end{equation}
Both (\ref{eq:UserDynamicProgrammingVReG}) and (\ref{eq:UserDynamicProgrammingVReB}) are linear equations on $\overline{V}(S_G,\alpha_G;\alpha_B,R)$ and $\overline{V}(S_B,\alpha_B;\alpha_G,R) $, thus, we can solve them together and achieve
\begin{equation}
\label{eq:UserWidehatJSG}
\begin{array}{l}
\overline{V}(S_G,\alpha_{G};\alpha_{B},R)  =\frac{ (1 - \delta p(S_B,\alpha_{B},  S_B) )l(S_G,\alpha_{G},R)  +   \delta p(S_G,\alpha_{G},S_B)   l(S_B,\alpha_{B},R)  }{I_p(\alpha_{G},\alpha_{B}) };
\end{array}
\end{equation}
\begin{equation}
\label{eq:UserWidehatJSB}
\begin{array}{l}
\overline{V}(S_B,\alpha_{B};\alpha_{G},R)  = \frac{ \delta p(S_B,\alpha_{B},  S_G)l(S_G,\alpha_{G},R)  +  ( 1 - \delta p(S_G,\alpha_{G},S_G)) l(S_B,\alpha_{B},R) }{I_p(\alpha_{G},\alpha_{B})},
\end{array}
\end{equation}
where
\begin{equation}
\label{eq:UserIp}
\begin{array}{c}
I_p(\alpha_{G},\alpha_{B}) =  \Big(1 - \delta p(S_G,\alpha_{G},S_G) \Big) \Big( 1- \delta p(S_B,\alpha_{B},S_B) \Big)   - \delta^2 p(S_G,\alpha_{G},S_B) p(S_B,\alpha_{B},S_G).
\end{array} 
\end{equation}
As a result, we can find $\pi_R^*$ by solving (\ref{eq:UserDynamicProgrammingPiReG}) and (\ref{eq:UserDynamicProgrammingPiReB}) with (\ref{eq:UserWidehatJSG}) and (\ref{eq:UserWidehatJSB}), respectively. Since there are only two protection choices $A_H$ and $A_L$, we can find the optimal protection policy by comparing the action dependent expected cumulative effective losses under $A_H$ and $A_L$.
\begin{lemma}
\label{lem:UserVKBSol2}
The optimal protection policy $\pi_R^*$ given the coverage level $R$ can be summarized as follows:
{ \small
\begin{itemize}
\item $\pi_R^* = \Pi_{LL}$ if and only if $\overline{V}(S_G,A_H;A_L,R)  \geq \overline{V}(S_G,A_L;A_L,R)$ and $\overline{V}(S_B,A_H;A_L,R)  \geq \overline{V}(S_B,A_L;A_L,R)$;
\item $\pi_R^* = \Pi_{LH}$ if and only if  $\overline{V}(S_G,A_H;A_H,R)  \geq \overline{V}(S_G,A_L;A_H,R)$ and $\overline{V}(S_B,A_H;A_L,R)  < \overline{V}(S_B,A_L;A_L,R)$;
\item $\pi_R^* = \Pi_{HL}$ if and only if $\overline{V}(S_G,A_H;A_L,R)  < \overline{V}(S_G,A_L;A_L,R)$ and $\overline{V}(S_B,A_H;A_H,R)  \geq \overline{V}(S_B,A_L;A_H,R)$;
\item $\pi_R^* = \Pi_{HH}$ if and only if $\overline{V}(S_G,A_H;A_H,R)  < \overline{V}(S_G,A_L;A_H,R)$ and $\overline{V}(S_B,A_H;A_H,R)  < \overline{V}(S_B,A_L;A_H,R)$.
\end{itemize}
}
\end{lemma}
\begin{proof}
The user chooses a protection policy with lower expected cumulative effective losses in both good state and bad state.
\end{proof}
We consider that the user always takes $A_L$ when $\overline{V}(s,A_H;\alpha_{s^c},R) = \overline{V}(s,A_L;\alpha_{s^c},R)$. We can further simplify the comparisons in Lemma \ref{lem:UserVKBSol2} as shown in the following proposition.
\begin{proposition}
\label{pro:UserH}
Let us define function $h:\mathcal{S}\times \mathcal{A} \times \mathcal{R}\rightarrow \mathbb{R}$ as
\[\begin{array}{l}
h(s,\alpha_{s^c},R) =  (1-R)\delta \left( p(s,A_H,s^c)   -   p(s,A_L,s^c) \right) ( X_{s^c} - X_s)   \\ \ \ \ \ \   +     (1-\delta + \delta p(S_B,\alpha_{s^c},S_G)   +   \delta p(S_G,\alpha_{s^c},S_B) )( C_H - C_L  ), 
\end{array} \]
the optimal protection policy $\pi_R^*$ can be summarized as follows:
\begin{itemize}
\item $\pi^* = \Pi_{LL}$ if and only if $h(S_G,A_L,R) \geq 0$ and $h(S_B,A_L,R) \geq 0 $;
\item $\pi^*=\Pi_{LH}$ if and only if $h(S_G,A_H,R) \geq 0$ and $h(S_B,A_L,R) < 0 $;
\item $\pi^*=\Pi_{HL}$ if and only if $h(S_G,A_L,R) < 0$ and $h(S_B,A_H,R) \geq 0 $;
\item $\pi^*=\Pi_{HH}$ if and only if $h(S_G,A_H,R) < 0$ and $h(S_B,A_H,R) < 0 $.
\end{itemize}
\end{proposition}
\begin{proof}
See Appendix A.
\end{proof}
Thus, we could obtain the optimal protection policy of the user by analyzing $h(s,\alpha_{s^c},R)$, and we further have the following observation on it. 

\begin{proposition}
\label{pro:UserHMonotonicity}
Function $h(s,\alpha_{s^c},R)$ is linearly increasing on the coverage level $R$. 
\end{proposition} 
\begin{proof}
We can see that $h(s,\alpha_{s^c},R)$ is linear on $R$ with a slope of $ \\ -\delta ( p(s,A_H,s^c)- p(s,A_L,s^c) ) \left( X_{s^c} - X_s \right) $. From the properties of protections and direct losses, we have $p(S_G,A_H,S_B)   -   p(S_G,A_L,S_B) < 0$, $X_B - X_G > 0$, $p(S_B,A_H,S_G)   -   p(S_B,A_L,S_G) > 0$, and $X_G - X_B < 0$. As a result, $ -\delta ( p(s,A_H,s^c)   -   p(s,A_L,s^c) ) ( X_{s^c} - X_s) > 0$ and $h(s,\alpha_{s^c},R)$ is linearly increasing on $R$. 
\end{proof}

Before we obtain the optimal protection policy $\pi_R^*$, we note the following proposition regarding the uniqueness of $\pi_R^*$. 
\begin{theorem}
\label{the:UserUniqueness}
The optimal protection policy $\pi_R^*$ is unique.
\end{theorem}
\begin{proof}
See Appendix B.
\end{proof}

With Lemma \ref{lem:UserVKBSol2}, Proposition \ref{pro:UserHMonotonicity}, and Theorem \ref{the:UserUniqueness}, we can obtain the optimal protection policies of the user with respect to the coverage level as stated in the following proposition. 
\begin{proposition}
\label{pro:UserSwitchingCoverageLevel}
Let us define the value of transition probabilities as 
\begin{equation}
\label{eq:UserValueTransitionProbabilities1}
\begin{array}{l}
\rho  = p(S_B,A_H,S_G) + p(S_G,A_H,S_B)  -  p(S_B,A_L,S_G)  - p(S_G,A_L,S_B).
\end{array}
\end{equation}
The user's optimal protection policies with respect to the insurer's coverage level can be summarized with the following cases as also shown in Fig. \ref{fig:NetworkUserAttackerInsurer}. 

\noindent\textbf{Case 1:} If $h(S_G,A_L,0)\geq 0$ and $h(S_B,A_L,0)\geq  0$, the optimal protection policies $\pi_R^* = \Pi_{LL}$ for $R\in [0,1]$. 

\noindent\textbf{Case 2:} If $h(S_G,A_L,0)<0$ and $h(S_B,A_H,0)\geq 0$, we have $\rho < 0$ in this case. The optimal protection policies $\pi_R^* = \Pi_{HL}$ for $R\in [0,R_G)$ and $\pi_R^* = \Pi_{LL}$ for $R\in [R_G, 1]$, where
\[\begin{array}{l}
R_G =  1 - \frac{  (1-\delta + \delta p(S_B,A_L,S_G)   +   \delta p(S_G,A_L,S_B) )( C_H - C_L  )    }{\delta( p(S_G,A_L,S_B)   -   p(S_G,A_H,S_B) ) ( X_B - X_G)} .
\end{array} \]

\noindent\textbf{Case 3:} If $h(S_G,A_H,0 ) \geq 0$ and $h(S_B,A_L,0) < 0$, we have $\rho > 0$  in this case. The optimal protection policies $\pi_R^* = \Pi_{LH}$ for $R\in [0,R_B)$ and $\pi_R^* = \Pi_{LL}$ for $R\in [R_B, 1]$, where
\[\begin{array}{l}
R_B = 1 - \frac{ (1-\delta + \delta p(S_G,A_L,S_B)   +   \delta p(S_B,A_L,S_G)  )( C_H - C_L  ) }{\delta( p(S_B,A_H,S_G)   -   p(S_B,A_L,S_G) ) ( X_B - X_G) }.
\end{array}\] 

\noindent\textbf{Case 4:} If $h(S_G,A_H,0 )  < 0$ and $h(S_B,A_H,0) < 0$, 
\begin{itemize}
\item Case 4(a): If $\rho < 0$, $\pi_R^* = \Pi_{HH}$ for $R\in [0,R_B)$, $\pi_R^* = \Pi_{HL}$ for $R\in [R_B,R_G)$, and $\pi_R^* = \Pi_{LL}$ for $R\in [R_G, 1]$, where
\[\begin{array}{l}
R_G =  1 - \frac{  (1-\delta + \delta p(S_B,A_L,S_G)   +   \delta p(S_G,A_L,S_B) )( C_H - C_L  )    }{\delta( p(S_G,A_L,S_B)   -   p(S_G,A_H,S_B) ) ( X_B - X_G)} , \end{array}  \]
\[\begin{array}{l}
 R_B = 1 - \frac{ (1-\delta + \delta p(S_G,A_H,S_B)   +   \delta p(S_B,A_H,S_G)  )( C_H - C_L  ) }{\delta( p(S_B,A_H,S_G)   -   p(S_B,A_L,S_G) ) ( X_B - X_G) } .
\end{array} \]
\item Case 4(b): If $\rho > 0$, $\pi_R^* = \Pi_{HH}$ for $R\in [0,R_G)$, $\pi_R^* = \Pi_{LH}$ for $R\in [R_G,R_B)$, and $\pi_R^* = \Pi_{LL}$ for $R\in [R_B, 1]$, where
\[\begin{array}{l}
R_G =  1 - \frac{  (1-\delta + \delta p(S_B,A_H,S_G)   +   \delta p(S_G,A_H,S_B) )( C_H - C_L  )    }{\delta( p(S_G,A_L,S_B)   -   p(S_G,A_H,S_B) ) ( X_B - X_G)} ,
\end{array} \]
\[\begin{array}{l}
R_B = 1 - \frac{ (1-\delta + \delta p(S_G,A_L,S_B)   +   \delta p(S_B,A_L,S_G)  )( C_H - C_L  ) }{\delta( p(S_B,A_H,S_G)   -   p(S_B,A_L,S_G) ) ( X_B - X_G) }.
\end{array} \]
\item Case 4(c): If $\rho = 0$, $\pi_R^* = \Pi_{HH}$ for $R\in [0, R_s)$ and $\pi_R^* = \Pi_{LL}$ for $R\in [R_s, 1]$, where
\[\begin{array}{l}
R_s = R_G  =  1 - \frac{  (1-\delta + \delta p(S_B,A_H,S_G)   +   \delta p(S_G,A_H,S_B) )( C_H - C_L  )    }{\delta( p(S_G,A_L,S_B)   -   p(S_G,A_H,S_B) ) ( X_B - X_G)}
\\ = R_B =  1 - \frac{ (1-\delta + \delta p(S_G,A_H,S_B)   +   \delta p(S_B,A_H,S_G)  )( C_H - C_L  ) }{\delta( p(S_B,A_H,S_G)   -   p(S_B,A_L,S_G) ) ( X_B - X_G) } 
\end{array}  \]
\end{itemize}
\end{proposition}
\begin{proof}
See Appendix C. 
\end{proof}

\begin{figure}[http]
\centering
\subfigure[Case 1]{
\includegraphics[width=0.3\textwidth]{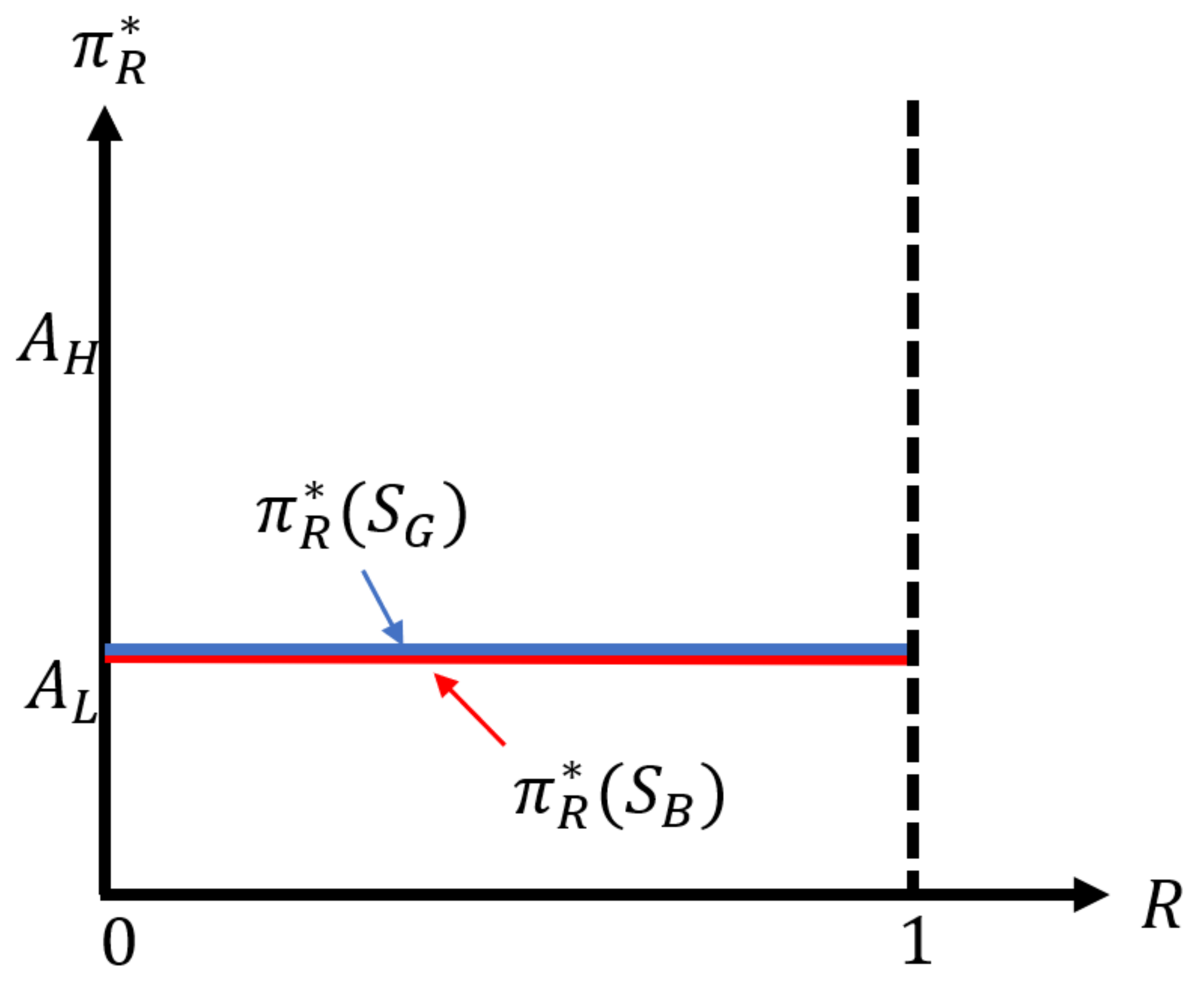}}
\subfigure[Case 2]{
\includegraphics[width=0.3\textwidth]{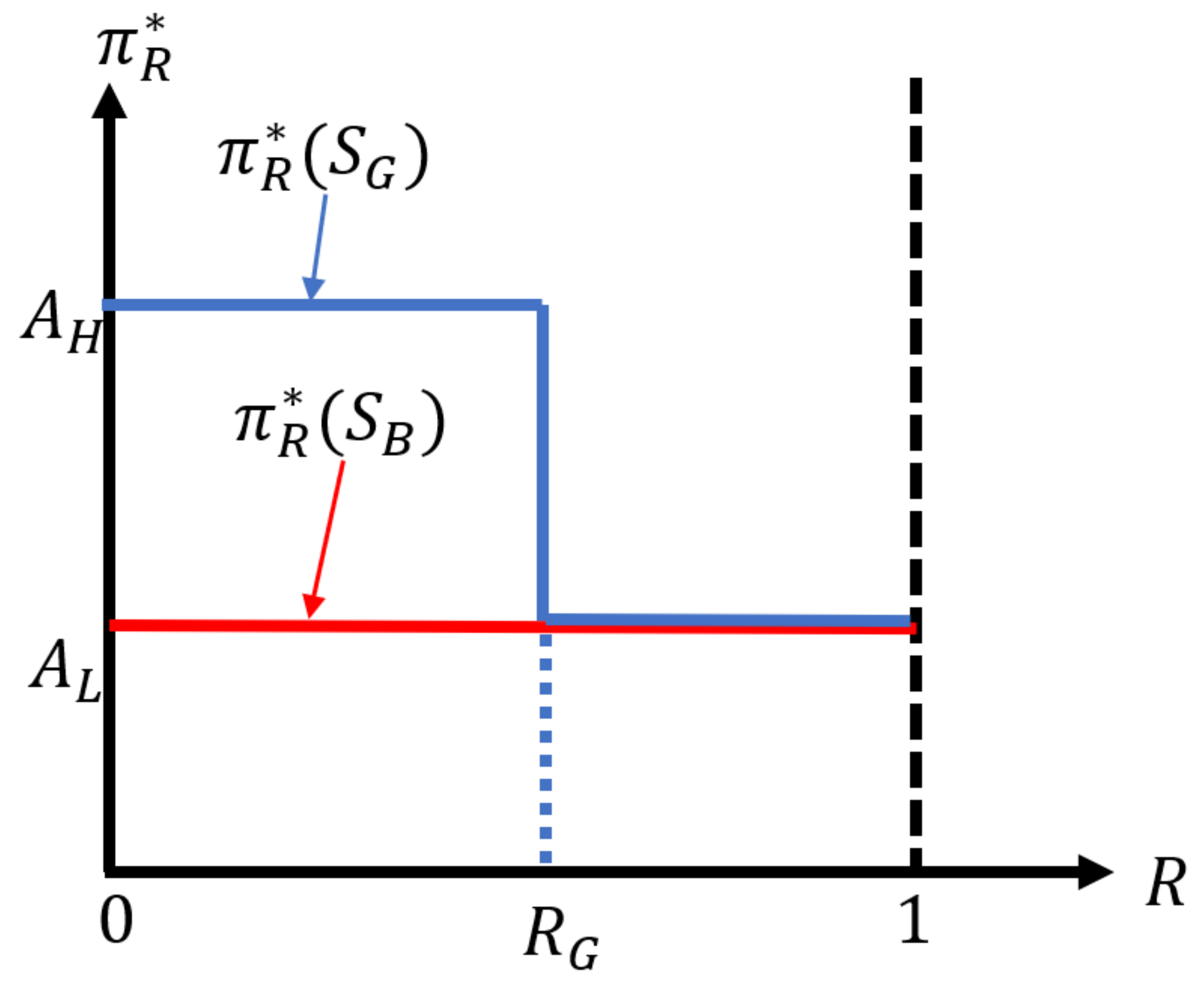}}
\subfigure[Case 3]{
\includegraphics[width=0.3\textwidth]{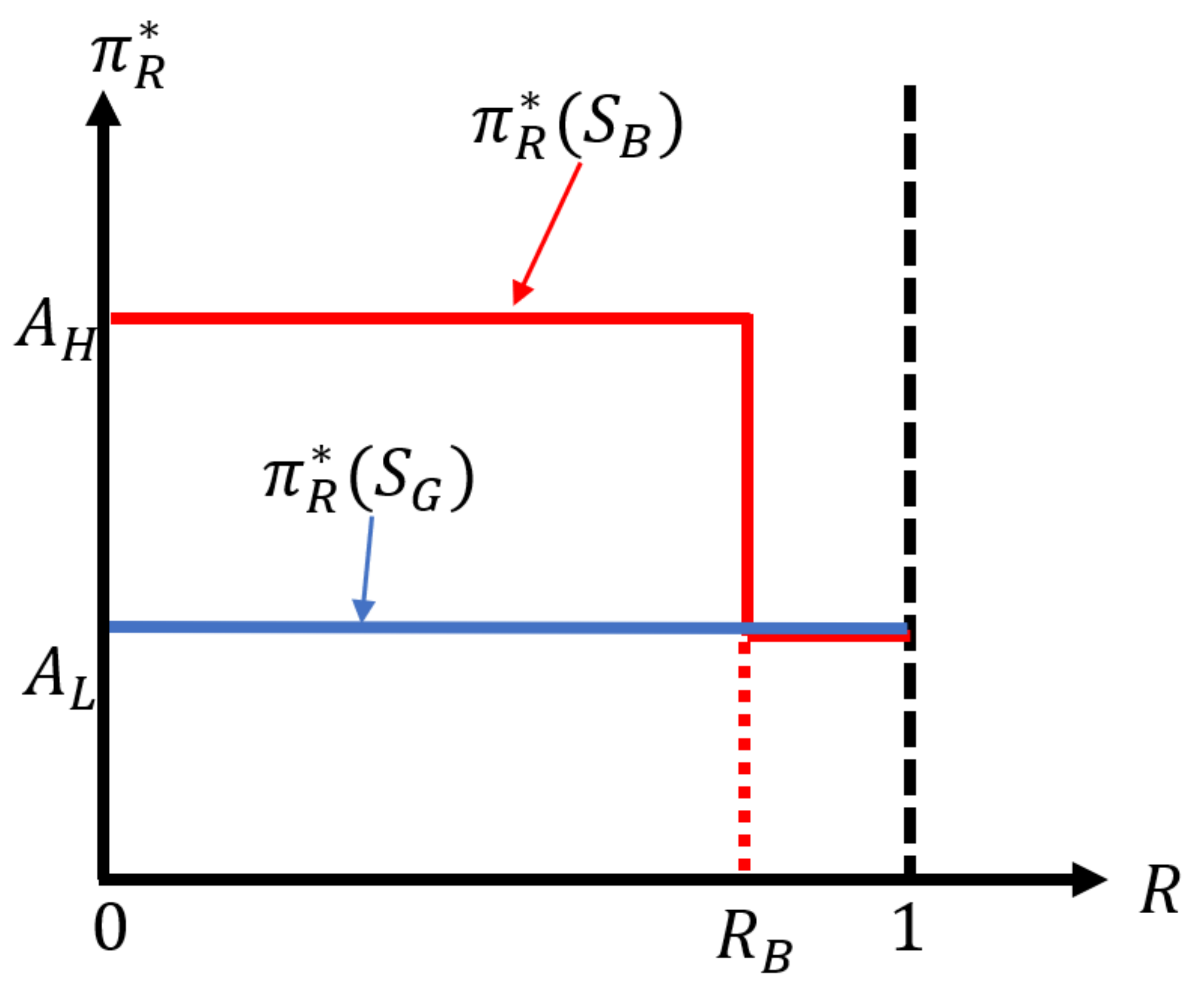}}
\subfigure[Case 4(a)]{
\includegraphics[width=0.3\textwidth]{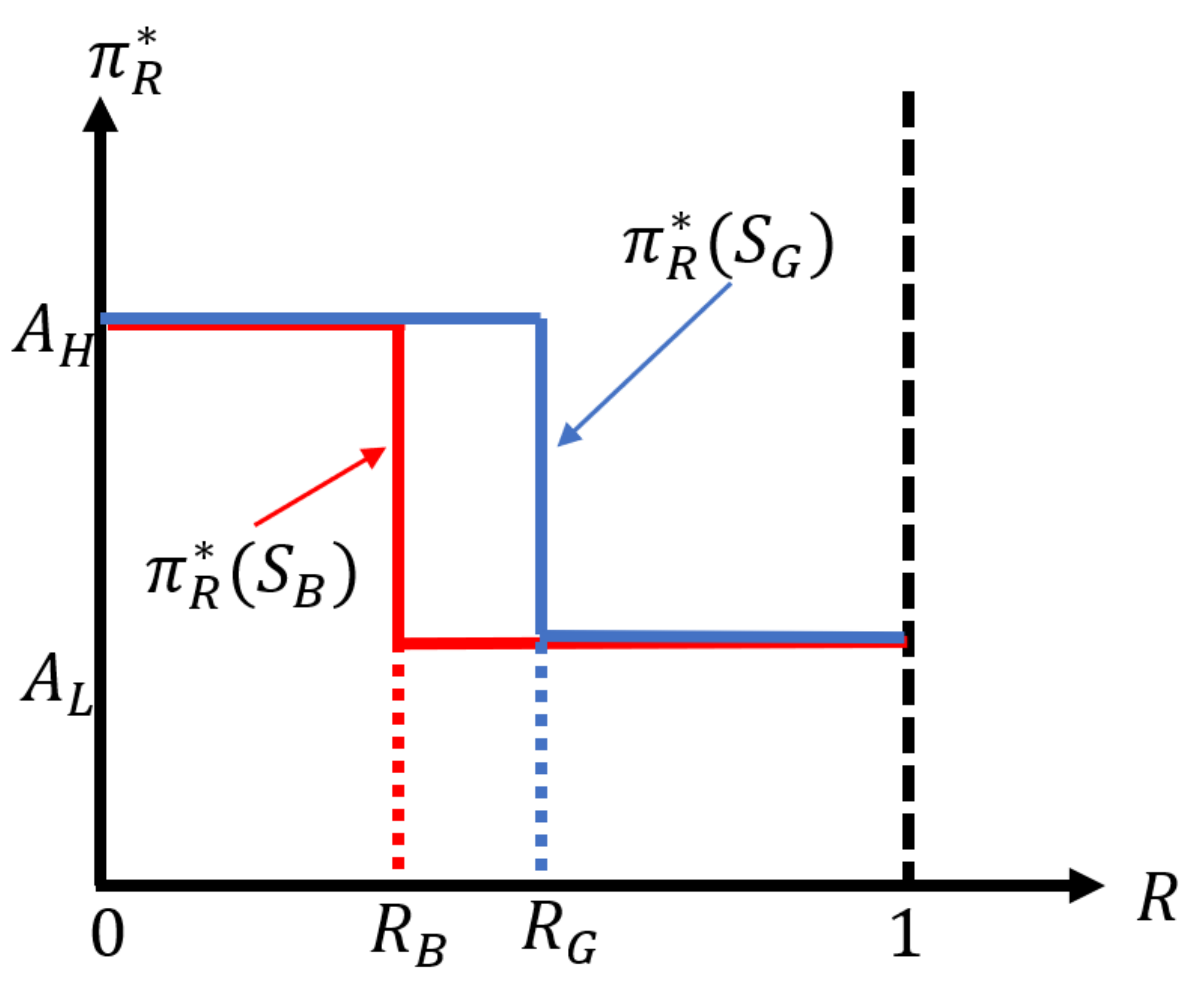}}
\subfigure[Case 4(b)]{
\includegraphics[width=0.3\textwidth]{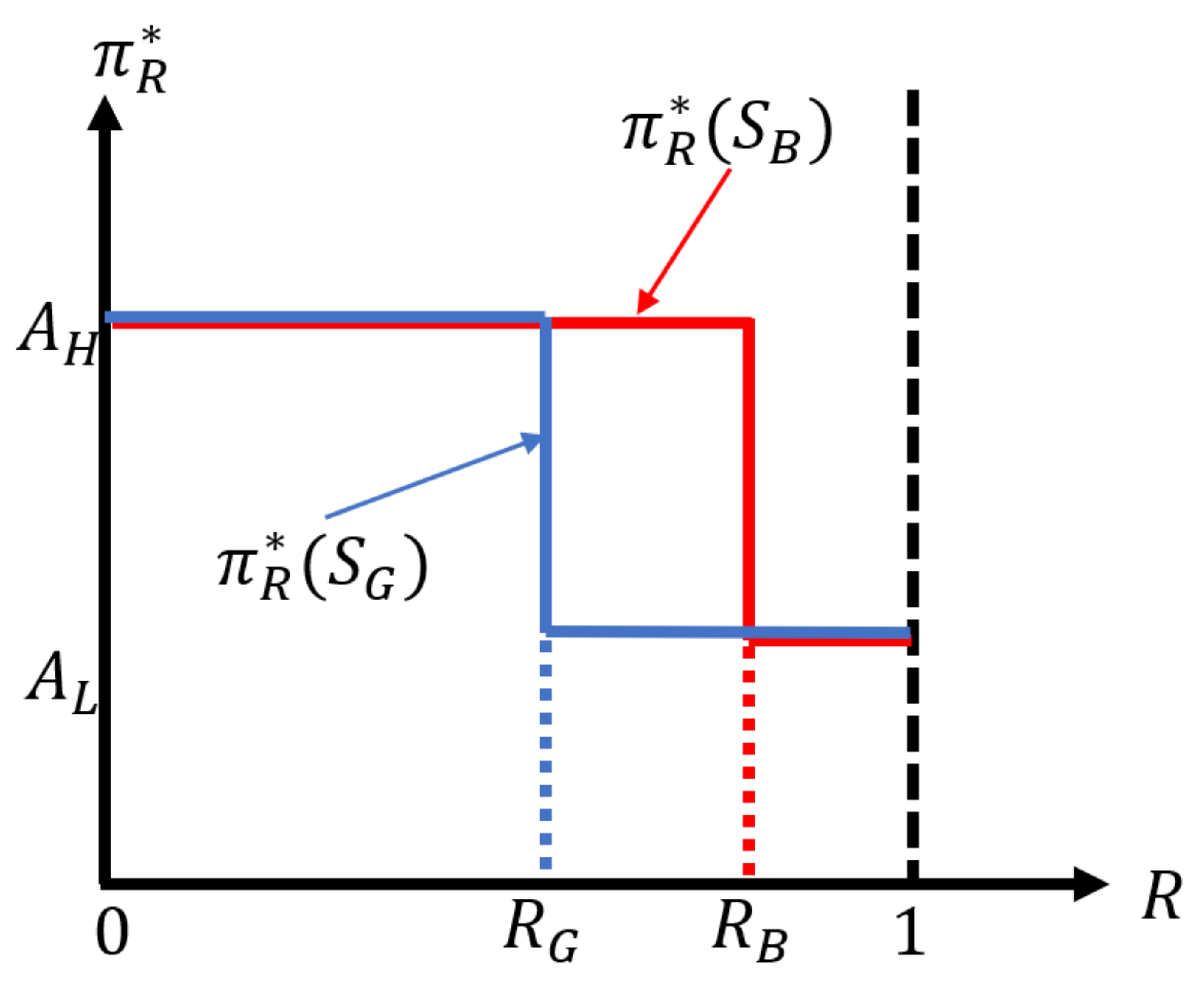}}
\subfigure[Case 4(c)]{
\includegraphics[width=0.3\textwidth]{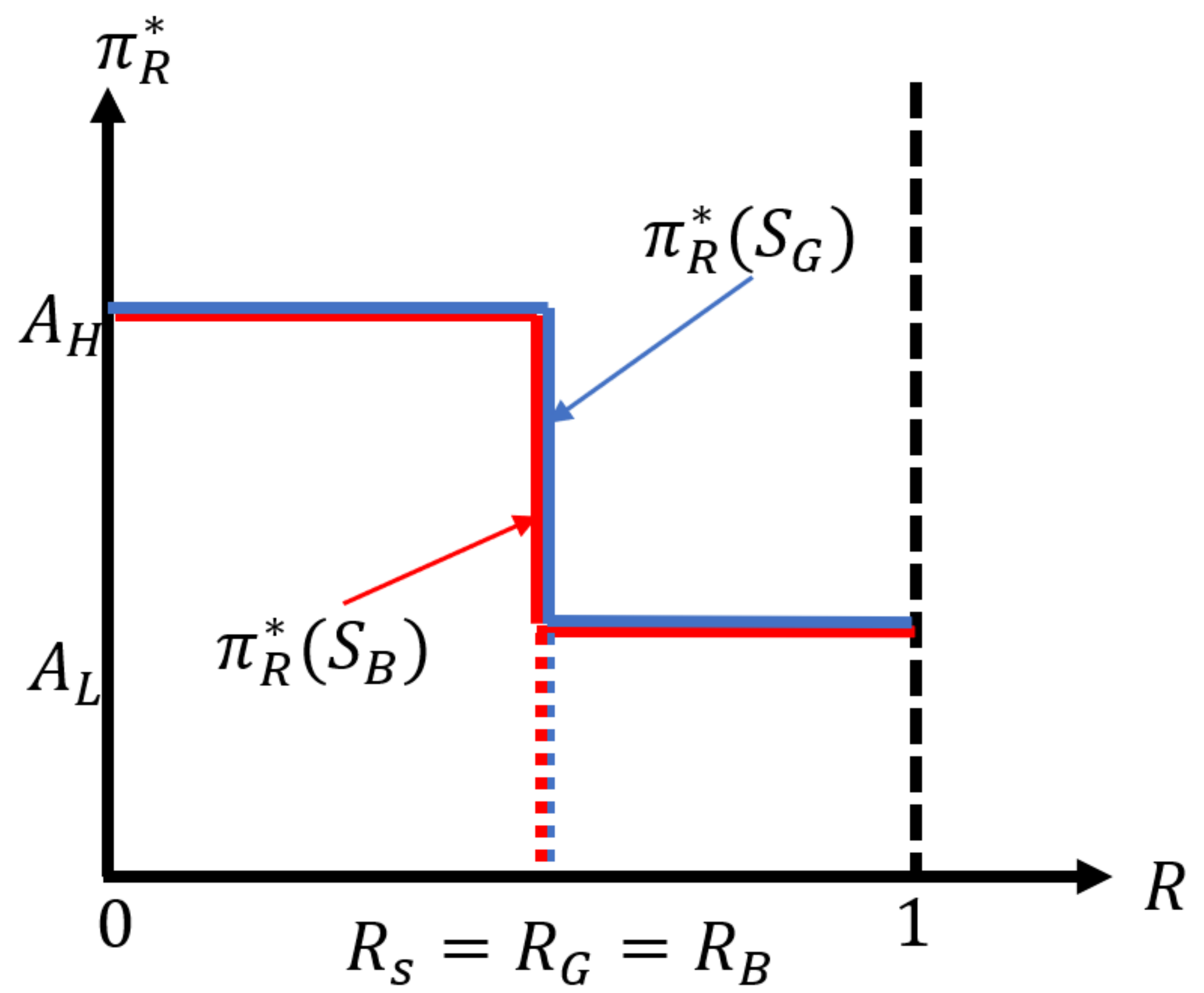}}
\caption{All possible cases of the user's optimal protection policies with respect to the insurer's coverage level. A detailed discussion is provided in Proposition \ref{pro:UserSwitchingCoverageLevel}.}
\label{fig:NetworkUserAttackerInsurer}
\end{figure}

We can see from Proposition \ref{pro:UserSwitchingCoverageLevel} that the user tends to take weak protections with the increase of the coverage level in all cases, and this reckless behavior is often referred as the risk compensation \cite{ewold1991insurance}. One critical impact of the risk compensation is the Peltzman effect as shown in the following theorem. 
\begin{theorem}
\label{the:UserPeltzmanEffect}
\textbf{Peltzman effect:} The user faces higher cyber risks under cyber-insurance. Such phenomena exists in the following cases:
\begin{itemize}
\item $R\in [R_G, 1]$ in Case 2 and Case 4(b);
\item $R\in [R_B, 1]$ in Case 3 and Case 4(a);
\item $R\in [R_s, 1]$ in Case 4(c). 
\end{itemize}
\end{theorem}
\begin{proof}
We only need to prove that the user has higher expected cumulative direct losses in these cases. Let 
$V_d(s, \pi)$ and $V_c(\pi)$ denote the expected cumulative direct losses and the expected cumulative costs given the initial state $s$ and the protection policy $\pi$. We have that $V_d(s, \pi)+ V_c(\pi)= V(s, \pi, 0)$. Recall that the optimal protection policy $\pi_0^*$ without insurance has $V(s,\pi_0^*, 0) \leq V(s, \pi, 0)$ for $\pi \in \Omega$, thus, when the user has a different optimal protection policy $\pi_R^* \neq \pi_0^*$ given the coverage level $R$, we have $V(s,\pi_0^*, 0) \leq V(s, \pi_R^*, 0)$. As a result, we can achieve $V_d(s, \pi_0^*)+ V_c(\pi_0^*) \leq V_d(s, \pi_R^*)+ V_c(\pi_R^*)$.  Note that $ V_c(\pi_0^*) > V_c(\pi_R^*)$ in these cases as $V_c(\Pi_{HH}) > V_c(\Pi_{HL}) > V_c(\Pi_{LL})$ and $V_c(\Pi_{HH}) > V_c(\Pi_{LH}) > V_c(\Pi_{LL})$ from $C_H > C_L$. Thus, we have $V_d(s, \pi_0^*) < V_d(s, \pi_R^*)$ and the user faces higher cyber risks. 
\end{proof}

\subsection{Optimal Insurance Contract}
\label{subsec:OIP}
Recall Section \ref{sec:Insurer}, the insurer's problem (\ref{eq:InsurerObjRe}) in this case can be written as follows:
\begin{equation}
\label{eq:InsurerPolicy}
\begin{array}{l}
\max\limits_{R\in  [0,1]} \tau(R)= V(S_G,\pi_{0}^*,0) - V(S_G,\pi_R^*,0)
\\
\text{s.t.} \ \ \tau(R) \geq 0,
\end{array}
\end{equation}
where $\tau(R)$ denotes the operating profit of the insurer if he provides a coverage level of $R$. Note that $S_G$ in $\tau(R)$ indicates that the initial state of the user is the good state. 
\begin{proposition}
\label{pro:InsurerOptimal}
The optimal insurance contract  $\{K^*,R^*\}$ for each case in Proposition \ref{pro:UserSwitchingCoverageLevel} can be summarized as follows:
\begin{itemize}
\item Case 1: $R^* \in [0, 1]$ and $K^* = R^* k(S_G,A_L;A_L)$,
\item Case 2: $R^*\in [0, R_G]$ and $K^* = R^* k(S_G,A_H;A_L)$, 
\item Case 3: $R^*\in [0, R_B]$ and $K^* = R^* k(S_G,A_L;A_H)$, 
\item Case 4(a): $R^*\in [0, R_B]$ and $K^* = R^* k(S_G,A_H;A_H)$,
\item Case 4(b): $R^*\in [0, R_G]$ and $K^* = R^* k(S_G,A_H;A_H)$,
\item Case 4(c): $R^*\in [0, R_s]$ and $K^* = R^* k(S_G,A_H;A_H)$,
\end{itemize}
where
\begin{equation}
\label{eq:UserKSG}
\begin{array}{l}
k(S_G,\alpha_{G};\alpha_{B}) =  \frac{ ( 1-\delta {p}(S_B,\alpha_{B},  S_B) )X_G  +   \delta {p}(S_G,\alpha_{G},S_B)   X_B }{I_p(\alpha_G, \alpha_B)}.
\end{array}
\end{equation}
For all the cases, the operating profit under the optimal insurance contract is $0$, i.e., $\tau(R^*) = 0$.
\end{proposition}
\begin{proof}
We can obtain the optimal insurance contract for all cases using the results from Proposition \ref{pro:InsurerOptimalGeneral}. In Case 1, any coverage level $R^*\in [0, 1]$ is optimal, and the associated premium $K^* =  V(S_G,\Pi_{LL},0) -  V(S_G,\Pi_{LL},R) = \overline{V}(S_G,A_L;A_L,0) -   \overline{V}(S_G,A_L;A_L,R) = R^*k(S_G, A_L; A_L)$, where $k(S_G, A_L; A_L)$ comes from (\ref{eq:UserK}) in Appendix A. Similarly, we can obtain the optimal insurance contracts in Cases 2, 3, and 4. 
\end{proof}
We can see from the optimal insurance contracts that the premium is linear on the coverage level, which can be summarized as the {linear insurance contract principle}. Moreover, all the optimal insurance contracts lead to a zero operating profit for the insurer, which indicates a {zero-operating profit principle}. The optimal insurance contracts usually provide limited coverage levels. When the coverage level is high, the user tends to act recklessly, which induces high risks and high direct losses of the user as shown in Theorem \ref{the:UserPeltzmanEffect}, and the insurer is required to cover the extra losses caused by that, which induces a negative profit of him. As a result, the insurer chooses not to provide the insurance to the user in that case. 
\section{Numerical Examples}
\label{sec:Num}
In this section, we first present numerical experiments on a two-state two-action user and a linear coverage insurer to verify our previous analytical results. We then present numerical experiments on a four-state three-action user with a linear coverage insurer and a threshold coverage insurer. 
\subsection{Two-State Two-Action User and Linear Coverage Insurer}
In this subsection, we aim to verify our analysis on the two-state two-action user and the linear coverage insurer with numerical experiments. We assume that the user has $\delta = 0.9$, $X_G = 0$, $X_B = 10$, $C_L = 0$, $C_H = 1$, $p(S_G,A_L,S_B) =p(S_G,A_L,S_G)= 0.5$, $p(S_B,A_L,S_G) =p(S_B,A_L,S_B)= 0.5$, $p(S_G,A_H,S_B) =1 - p(S_G,A_H,S_G)= 0.2$, and $p(S_B,A_H,S_G) =1- p(S_B,A_H,S_B)= 0.6$. We can achieve that $\rho= -0.20$, $h(S_G,A_H,0) = -1.88$, $h(S_G, A_L, 0) = -1.70$, $h(S_B,A_H, 0) = -0.08$, and $h(S_B, A_L, 0) = 0.10$. Thus, the user's optimal protection policies can be described as the Case 4(a) in Proposition \ref{pro:UserSwitchingCoverageLevel}. The optimal insurance contract $\{K^*,R^*\} $ has $R^*\in [0,R_B]$ and $K^*=R^*k(S_G,A_H;A_H)$, where $R_B=0.0889$ and $k(S_G,A_H;A_H)=21.9512$ from Proposition \ref{pro:InsurerOptimal}.

With the dynamic programming approach or linear programming approach in Section \ref{sec:User}, we can compute the optimal protection policies and the expected cumulative effective losses of the user as shown in Figs.\ref{fig:S2A2Linear}(a,b). We can further calculate the premium and the operating profit of the insurer as shown in Figs.\ref{fig:S2A2Linear}(c,d). We can see that the numerical results coincide with our analytical results. 

\begin{figure}[http]
\centering
\subfigure[]{
\includegraphics[width=0.44\textwidth]{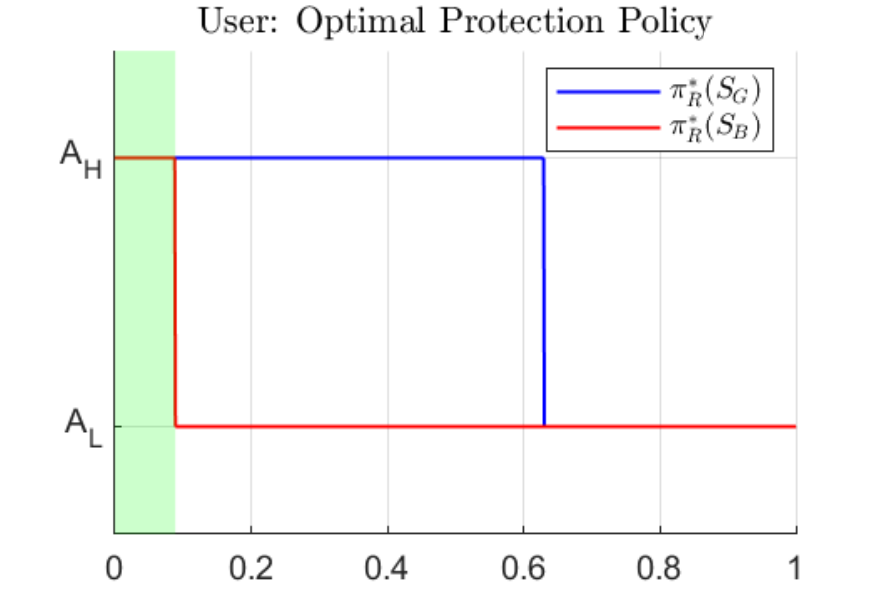}}
\subfigure[]{
\includegraphics[width=0.44\textwidth]{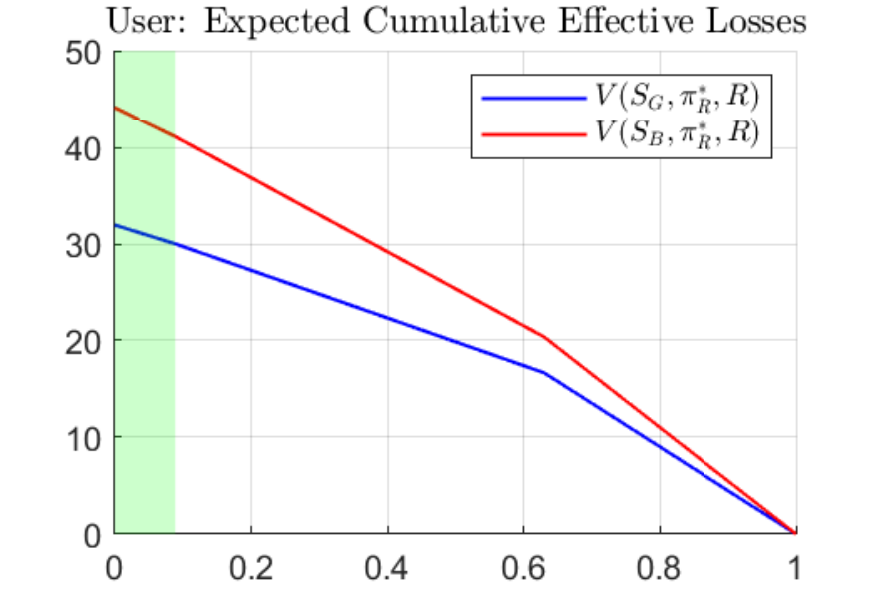}}
\subfigure[]{
\includegraphics[width=0.44\textwidth]{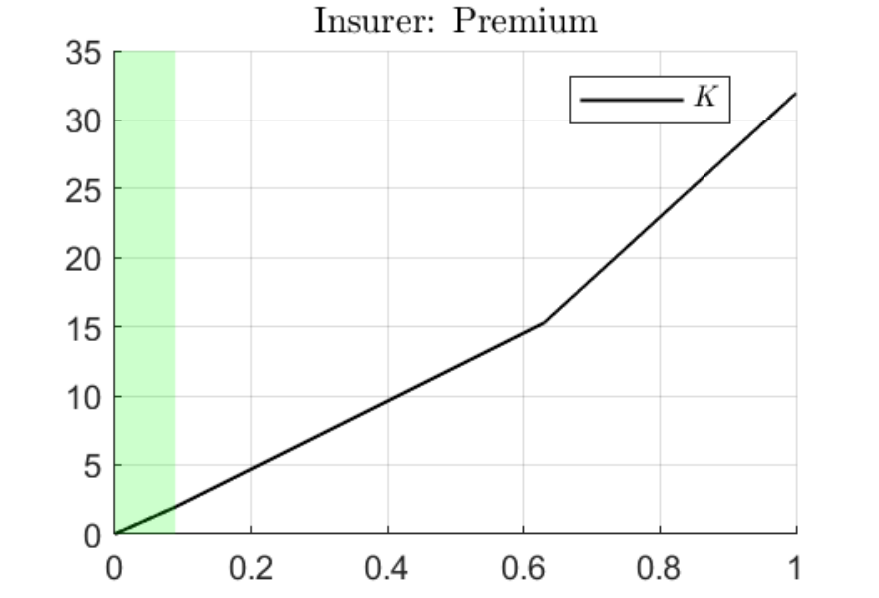}}
\subfigure[]{
\includegraphics[width=0.44\textwidth]{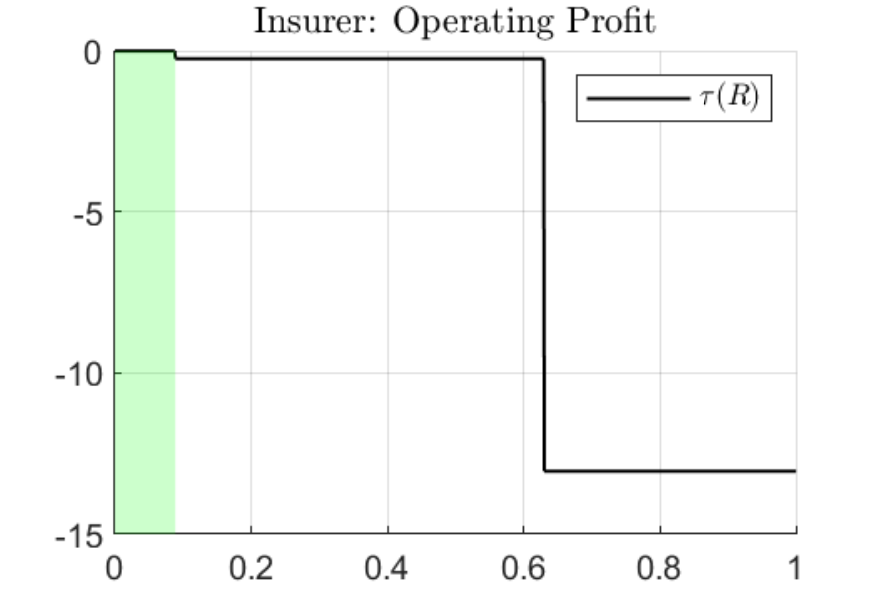}}
\caption{Two-State Two-Action User and Threshold Coverage Insurer. The horizontal axis in Figs. (a), (b), (c), (d) represents the coverage level $R$. The green area denotes the region of optimal insurance contracts.}
\label{fig:S2A2Linear}
\end{figure}

\subsection{Four-State Three-Action User and Linear Coverage Insurer}
In this subsection, we consider a more complicated example where the user has four states and three actions, and the insurer provides linear coverage. We show that our model can be used to analyze the interactions between the user and the insurer in a numerical way.  

We assume that the user's states can be identified as $S_G$, $S_{B,1}$, $S_{B,2}$, and $S_{B,3}$ with the state losses $X_G = 0$, $X_{B,1} = 4$, $X_{B,2}=8$, and $X_{B,3}=16$, respectively. $S_G$ indicates the good state, while $S_{B,i}$ indicates the bad states with $i$ capturing the level of the damage. The user can take no protection $A_0$, weak protection $A_L$, or strong protection $A_H$, and the costs of them can be identified as $c(A_0) = 0$, $c(A_L)=0.3$, and $c(A_H) = 0.6$, respectively. Different actions have different impacts on the transition probabilities. For convenience, we summarize the transition probabilities by the following matrix:
\begin{equation}
\label{eq:NEP}
P_{a} = 
 \left[  \begin{array}{cccc}
{p(S_G,a,S_G)}&{p(S_G,a,S_{B,1})}&{p(S_G,a,S_{B,2})}&{p(S_G,a,S_{B,3})} \\ 
{p(S_{B,1},a,S_G)}&{p(S_{B,1},a,S_{B,1})}&{p(S_{B,1},a,S_{B,2})}&{p(S_{B,1},a,S_{B,3})} \\ 
{p(S_{B,2},a,S_G)}&{p(S_{B,2},a,S_{B,1})}&{p(S_{B,2},a,S_{B,2})}&{p(S_{B,2},a,S_{B,3})} \\ 
{p(S_{B,3},a,S_G)}&{p(S_{B,3},a,S_{B,1})}&{p(S_{B,3},a,S_{B,2})}&{p(S_{B,3},a,S_{B,3})} \\ 
\end{array}  \right]
\end{equation}
The user has a larger probability of going to the good state and a smaller probability of going to the bad state with better protections. We then take the following transition probabilities in this example:
\[ P_{A_0} = \left[  \begin{array}{cccc}
{0.25}&{0.25}&{0.25}&{0.25} \\ 
{0.25}&{0.25}&{0.25}&{0.25} \\
{0.25}&{0.25}&{0.25}&{0.25} \\
{0.25}&{0.25}&{0.25}&{0.25} \\
\end{array}  \right]; \]
\[ P_{A_L} = \left[  \begin{array}{cccc}
{0.4}&{0.3}&{0.2}&{0.1} \\ 
{0.4}&{0.3}&{0.2}&{0.1} \\ 
{0.4}&{0.3}&{0.2}&{0.1} \\ 
{0.4}&{0.3}&{0.2}&{0.1} \\ 
\end{array}  \right]; \]
\[ P_{A_H} = \left[  \begin{array}{cccc}
{0.8}&{0.2}&{0.0}&{0.0} \\ 
{0.7}&{0.2}&{0.1}&{0.0} \\
{0.6}&{0.2}&{0.1}&{0.1} \\
{0.5}&{0.2}&{0.2}&{0.1} \\
\end{array}  \right]. \]
Let $\delta = 0.9$, the optimal protection policies and the expected cumulative effective losses of the user are shown in Figs.\ref{fig:S4A3Linear}(a,b). We can see from (a,b) that the user decreases his protections with the increase of the coverage level, and the user also has lower expected cumulative effective losses with higher coverage levels. The premium and the operating profit of the insurer are shown in Figs.\ref{fig:S4A3Linear}(c,d). We can see that with the increase of the coverage level, the premium is linearly increasing. Moreover, the maximum operating profit that can be achieved by the insurer is $0$. We can also observe that the optimal insurance contract tends to provide limited coverage levels, and higher coverage levels can lead to negative operating profits of the insurer. Thus, the risk compensation, the zero-operating profit principle, and the linear insurance contract principle still hold in this example. 
\begin{figure}[http]
\centering
\subfigure[]{
\includegraphics[width=0.44\textwidth]{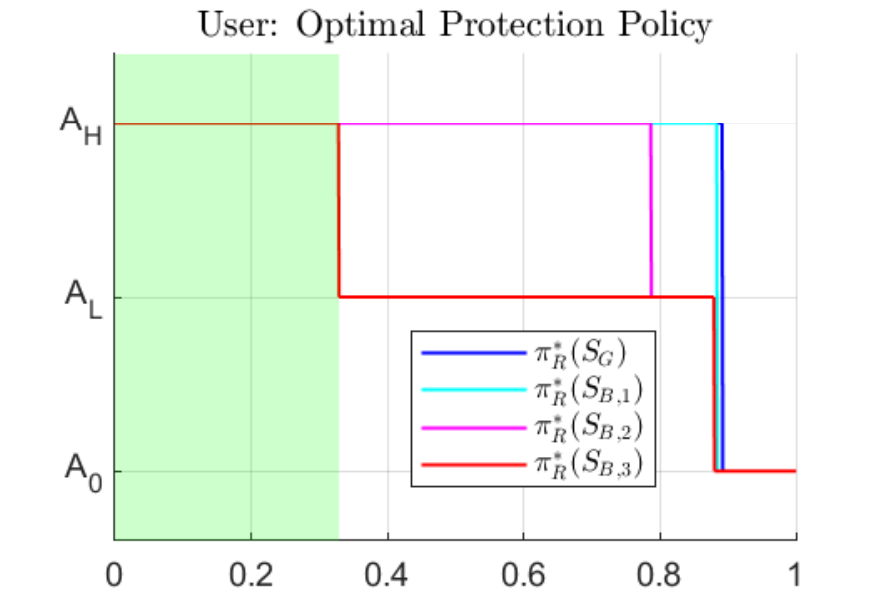}}
\subfigure[]{
\includegraphics[width=0.44\textwidth]{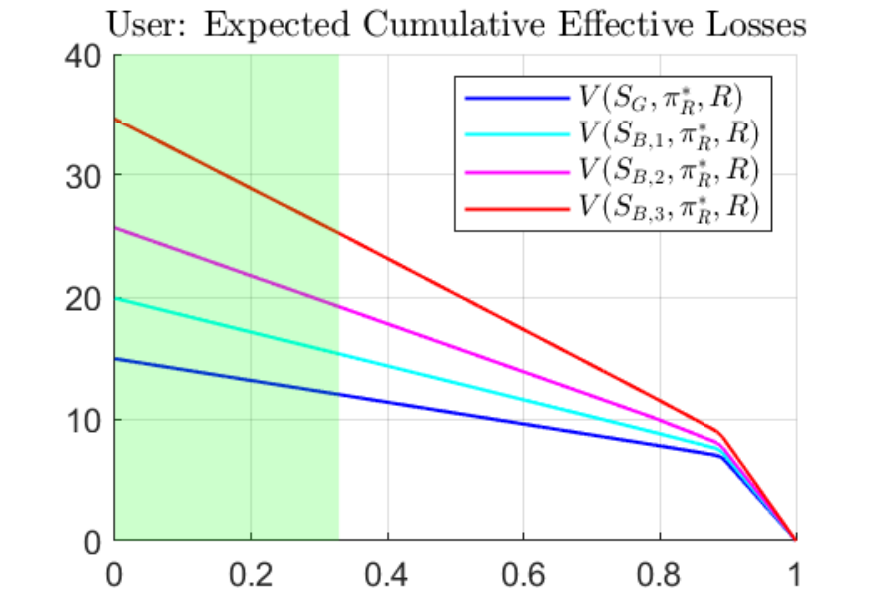}}
\subfigure[]{
\includegraphics[width=0.44\textwidth]{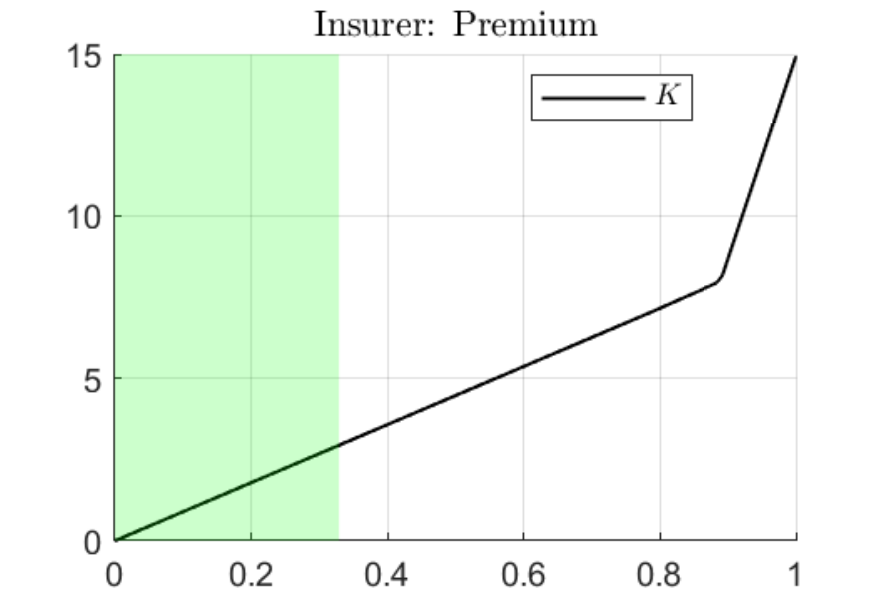}}
\subfigure[]{
\includegraphics[width=0.44\textwidth]{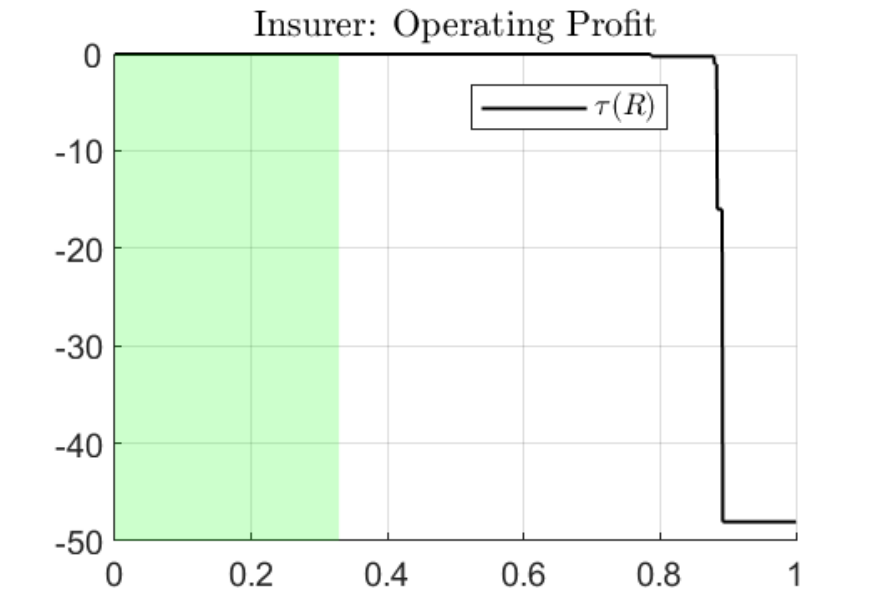}}
\caption{Four-State Three-Action User and Linear Coverage Insurer. The horizontal axis in Figs. (a), (b), (c), (d) represents the coverage level $R$. The green area denotes the region of optimal insurance contracts.}
\label{fig:S4A3Linear}
\end{figure}

\subsection{Four-State Three-Action User and Threshold Coverage Insurer}
In this subsection, we consider a threshold coverage insurance and show its impact on the four-state three-action user and the insurer. 

We use the same settings for the user as in the previous subsection. The threshold insurance contract has two coverage levels $R_0 = 0$ and $R_1 = 0.9$, which are distinguished by a threshold $X_R \in [0, 20]$. When the loss of the user $x \leq X_R$, the insurer provides no coverage $R_0x$, otherwise, the insurer provides a coverage $R_1x$. A lower $X_R$ indicates that the insurance has a higher coverage for smaller losses. The objective of the insurer is to maximize his operating profit by finding the optimal threshold $X_R^*$ and the associated premium $K^*$. 

The optimal protection policies and the expected cumulative effective losses of the user are shown in Fig.\ref{fig:S4A3Jump}(a,b), and we can see from them that the user decreases his protections with the decrease of the threshold $X_R$, which indicates that the user tends to act recklessly knowing that the insurer provides a high coverage even he has a small loss from cyber risks. Moreover, we can see that the premium is a staircase function on the threshold $X_R$, and it decreases with the increase of $X_R$, which shows that the insurer charges a higher premium to provide a higher coverage level. The maximum operating profit that can be achieved by the insurer is $0$. As a result, this example shows the similar risk compensation and zero-operating profit principle as in the previous examples. Note that the gray area has $X_R > X_{B, 3}$, i.e., the insurer provides no coverage for the user at any states, which is equivalent to the case when there is no insurance. 

\begin{figure}[http]
\centering
\subfigure[]{
\includegraphics[width=0.44\textwidth]{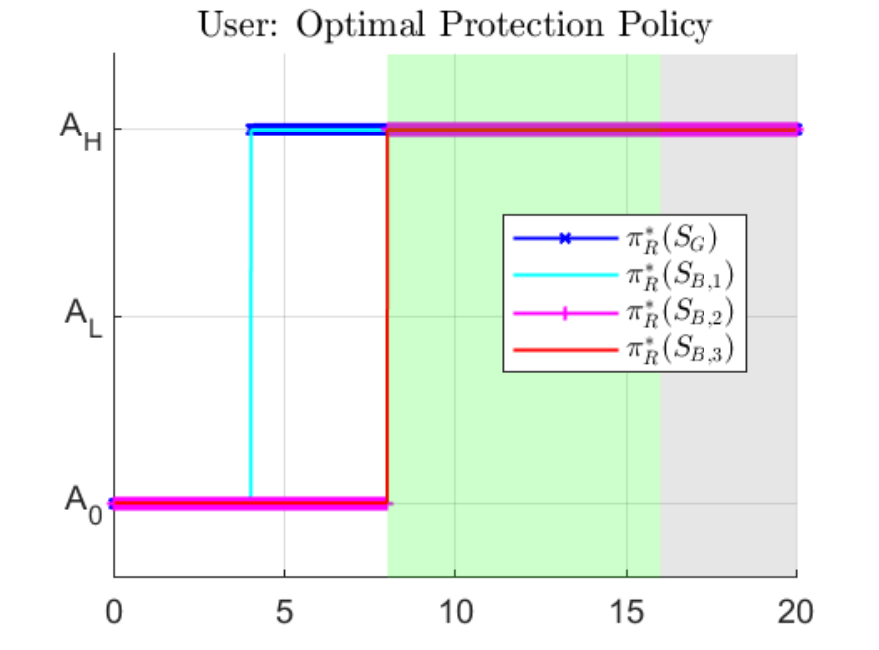}}
\subfigure[]{
\includegraphics[width=0.44\textwidth]{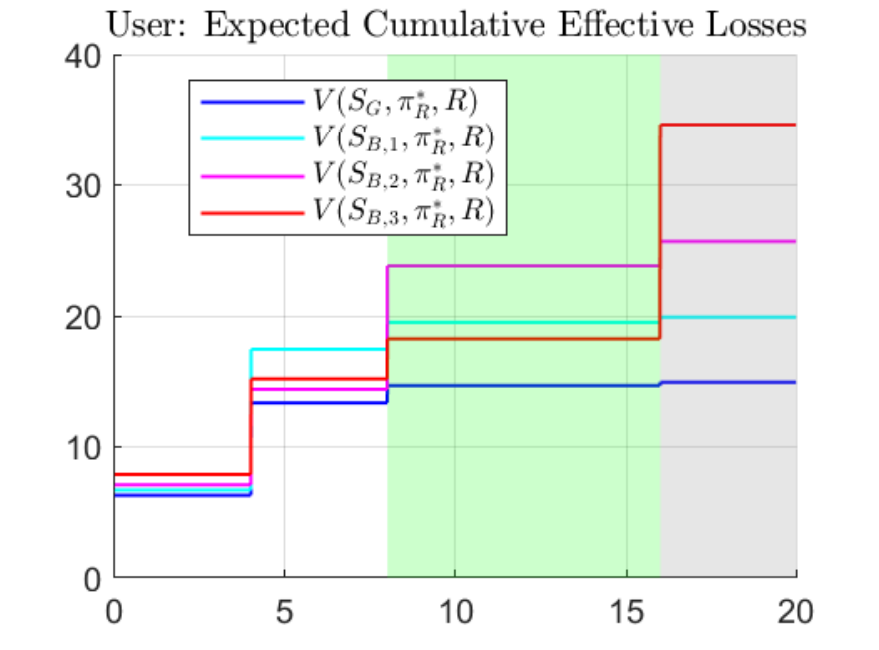}}
\subfigure[]{
\includegraphics[width=0.44\textwidth]{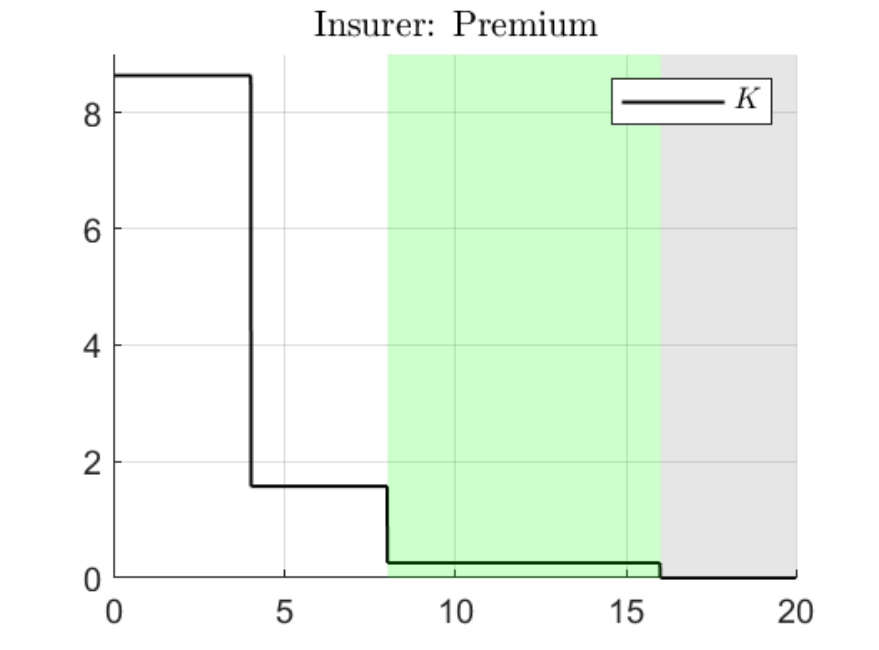}}
\subfigure[]{
\includegraphics[width=0.44\textwidth]{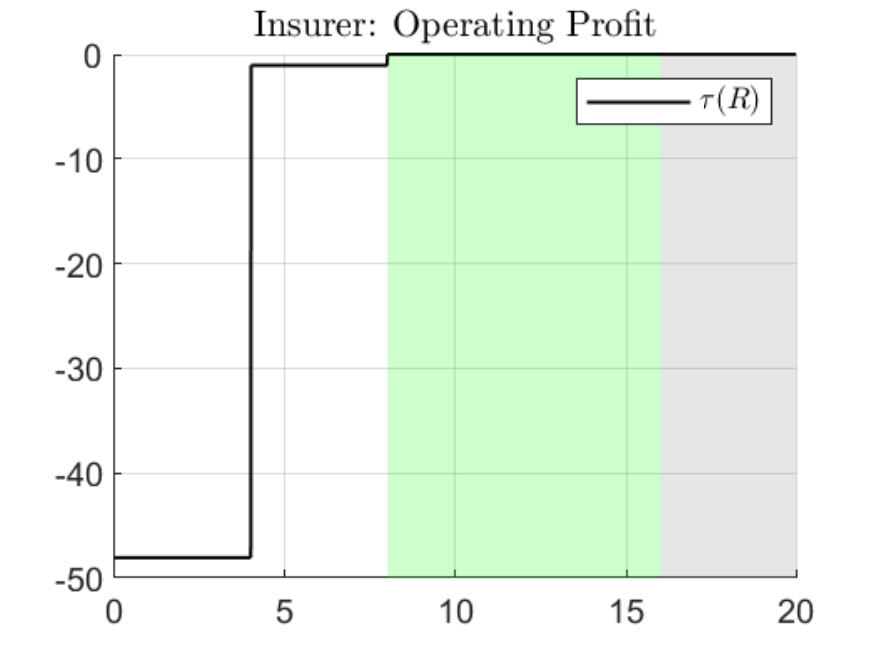}}
\caption{Four-State Three-Action User and Threshold Coverage Insurer. The horizontal axis in Figs. (a), (b), (c), (d) represents the threshold $X_R$. The green area denotes the region of optimal insurance contracts. }
\label{fig:S4A3Jump}
\end{figure}

\section{Conclusion}
\label{sec:Con}
In this paper, we have presented a dynamic moral-hazard type of principal-agent model to study the cyber-insurance and its impacts on the cyber-security. The dynamics and correlations of the cyber risks have been modeled by Markov decision processes where the user aims to find the optimal protection policy to mitigate the impacts of cyberattacks. Both the computational and analytical tools have been presented to design the optimal cyber-insurance contracts. We have studied and fully analyzed a case where the user has two states and two actions, and the insurer provides linear coverage insurance. We have further demonstrated the Peltzman effect that the user has higher cyber risks under insurance due to risk compensation, i.e., the user tends to act more recklessly knowing he is protected. We have presented the linear insurance contract principle and the zero-operating profit principle of the optimal cyber-insurance contract. Numerical experiments have been used to corroborate our results and further demonstrate the case study with a four-state three-action user and his interactions with linear coverage insurance and threshold coverage insurance. The risk compensation and the zero-operating profit principle have been shown to hold in these cases. One direction of future research is the investigation of cyber-insurance contracts over complex networks such as scale-free and small-world networks with dynamic cyber risks.

\section*{Appendix A. Proof of Proposition \ref{pro:UserH}}
To simplify the notation in this proof, we define the discounted transition probabilities as 
\[\widehat{p}(s,\alpha_s,s) = 1 - \delta  p(s,\alpha_s,s), \]
\[\widehat{p}(s,\alpha_s,s^c) = \delta  p(s,\alpha_s,s^c).\]
\begin{remark}
\label{rem:UserDiscountedPro}
The following facts hold for $\widehat{p}$:
\begin{itemize}
\item[(i)] $\widehat{p}(S_G,\alpha_{G},S_G)  - \widehat{p}(S_G,\alpha_{G},S_B) = \widehat{p}(S_B,\alpha_{B},S_B) -\widehat{p}(S_B,\alpha_{B},S_G) = 1-\delta $;
\item[(ii)] If $\delta = 1$, we have  $\widehat{p}(S_G,\alpha_{G},S_G) = \widehat{p}(S_G,\alpha_{G},S_B)$ and $\widehat{p}(S_B,\alpha_{B},S_B) = \widehat{p}(S_B,\alpha_{B},S_G)$.
\end{itemize}
\end{remark}
Thus, (\ref{eq:UserIp}) can be written as
\[\begin{array}{l}
I_p(\alpha_{G},\alpha_{B}) = \widehat{p}(S_G,\alpha_{G},S_G)\widehat{p}(S_B,\alpha_{B},S_B) - \widehat{p}(S_G,\alpha_{G},S_B) \widehat{p}(S_B,\alpha_{B},S_G). 
\end{array} \]
\begin{remark}
\label{rem:UserEigenvalue}
The following facts hold for $I_p$:
\begin{itemize}
\item[(i)] $I_p(\alpha_{G},\alpha_{B})  = (1-\delta +\widehat{p}(S_G,\alpha_{G},S_B) )(1-\delta + \widehat{p}(S_B,\alpha_{B},S_G))    - \widehat{p}(S_G,\alpha_{G},S_B) \widehat{p}(S_B,\alpha_{B},S_G)   = (1-\delta)^2+(1-\delta) (    \widehat{p}(S_G,\alpha_{G},S_B) +  \widehat{p}(S_B,\alpha_{B},S_G))  > 0$ when $0 \leq \delta < 1$;
\item[(ii)] $I_p(A_H,\alpha_{B}) - I_p(A_L,\alpha_{B}) =  (1-\delta)(  \widehat{p}(S_G,A_H,S_B) -   \widehat{p}(S_G,A_L,S_B) ) $;
\item[(iii)] $I_p(\alpha_{G},A_H) -  I_p(\alpha_{G},A_L)  = (1-\delta) (    \widehat{p}(S_B,A_H,S_G)- \widehat{p}(S_B,A_L,S_G))  $.
\end{itemize}
\end{remark}
The action dependent expected cumulative effective losses (\ref{eq:UserWidehatJSG}) and (\ref{eq:UserWidehatJSB}) can be rewritten as follows:
\begin{equation}
\label{eq:UserWidehatJRe}
\overline{V}(s,\alpha_{s};\alpha_{s^c},R) =  (1-R) k(s,\alpha_{s};\alpha_{s^c})  +  b(s,\alpha_{s};\alpha_{s^c}),
\end{equation}
where
\begin{equation}
\label{eq:UserK}
\begin{array}{c}
k(s,\alpha_{s};\alpha_{s^c}) =   \frac{  \widehat{p}(S_B,\alpha_{B},  s^c) X_G  +    \widehat{p}(S_G,\alpha_{G},s^c)   X_B }{ I_p(\alpha_{G},\alpha_{B})   } ;
\end{array}
\end{equation}
\begin{equation}
\label{eq:UserB}
\begin{array}{c}
b(s,\alpha_{G};\alpha_{s^c})  =   \frac{ \widehat{p}(S_B,\alpha_{B},  s^c) c(\alpha_{G})  +    \widehat{p}(S_G,\alpha_{G}, s^c)   c(\alpha_{B}) }{I_p(\alpha_{G},\alpha_{B}) }.
\end{array}
\end{equation}
Note that
\begin{equation}
\label{eq:UserVKSG}
\begin{array}{l}
k(S_G,A_H;\alpha_B) - k(S_G,A_L;\alpha_B) \\ = \frac{ \widehat{p}(S_B,\alpha_{B},S_B) X_G +  \widehat{p}(S_G,A_H,S_B) X_B }{I_p(A_H,\alpha_{B})}  - \frac{ \widehat{p}(S_B,\alpha_{B},S_B) X_G +  \widehat{p}(S_G,A_L,S_B) X_B }{I_p(A_L,\alpha_{B})} 

\\ = \frac{ \widehat{p}(S_B,\alpha_{B},S_B) I_p(A_L,\alpha_{B}) X_G +  \widehat{p}(S_G,A_H,S_B)I_p(A_L,\alpha_{B})  X_B   }{I_p(A_H,\alpha_{B}) I_p(A_L,\alpha_{B})}
  - \frac{   \widehat{p}(S_B,\alpha_{B},S_B) I_p(A_H,\alpha_{B}) X_G + \widehat{p}(S_G,A_L,S_B) I_p(A_H,\alpha_{B}) X_B            }{I_p(A_H,\alpha_{B}) I_p(A_L,\alpha_{B})}

\\ = \frac{ \widehat{p}(S_B,\alpha_{B},S_B) ( I_p(A_L,\alpha_{B}) - I_p(A_H,\alpha_{B})   )X_G  }{I_p(A_H,\alpha_{B}) I_p(A_L,\alpha_{B})}
 + \frac{ \widehat{p}(S_G,A_H,S_B)I_p(A_L,\alpha_{B})  X_B  -  \widehat{p}(S_G,A_L,S_B) I_p(A_H,\alpha_{B})  X_B            }{I_p(A_H,\alpha_{B}) I_p(A_L,\alpha_{B})} 

\\ = \frac{ (1-\delta) \widehat{p}(S_B,\alpha_{B},S_B) (\widehat{p}(S_G,A_L,S_B) - \widehat{p}(S_G,A_H,S_B)  )  X_G}{I_p(A_H,\alpha_{B}) I_p(A_L,\alpha_{B})}
 + \frac{  (1-\delta) \widehat{p}(S_B,\alpha_{B},S_B) \left(  \widehat{p}(S_G,A_H,S_B) - \widehat{p}(S_G,A_L,S_B)   \right)   X_B }{I_p(A_H,\alpha_{B}) I_p(A_L,\alpha_{B})}

\\ = \frac{ (1-\delta) \widehat{p}(S_B,\alpha_{B},S_B) \left(\widehat{p}(S_G,A_L,S_B) - \widehat{p}(S_G,A_H,S_B)  \right) \left(  X_G  - X_B \right)}{I_p(A_H,\alpha_{B}) I_p(A_L,\alpha_{B})}, 
\end{array}
\end{equation}
where the fourth equality is achieved by plugging Remark \ref{rem:UserEigenvalue}(i)(ii). Similarly, we can achieve that
 \begin{equation}
\label{eq:UserVKSB}
\begin{array}{l}
k(S_B,A_H;\alpha_{G}) - k(S_B,A_L;\alpha_{G})
 = \frac{(1-\delta)\widehat{p}(S_G,\alpha_{G},S_G)\left(   \widehat{p}(S_B,A_H,S_G) - \widehat{p}(S_B,A_L,S_G)   \right)  \left(  X_G  - X_B \right)}{I_p(\alpha_{G},A_H) I_p(\alpha_{G},A_L)  };
\end{array}
\end{equation}
\begin{equation}
\label{eq:UserVBSG}
\begin{array}{l}
b(S_G,A_H;\alpha_B) - b(S_G,A_L;\alpha_{B})
 =\frac{(1-\delta)\widehat{p}(S_B,\alpha_B,S_B) \left( \widehat{p}(S_B,\alpha_B,S_B) + \widehat{p}(S_G,\alpha_B,S_B)   \right)(C_H-C_L)}{I_p(A_H,\alpha_B)I_p(A_L,\alpha_B)};
\end{array}
\end{equation}
\begin{equation}
\label{eq:UserVBSB}
\begin{array}{l}
b(S_B,A_H;\alpha_{G}) - b(S_B,A_L;\alpha_{G}) 
 =  \frac{ (1-\delta)\widehat{p}(S_G,\alpha_{G},S_G) \left( \widehat{p}(S_G,\alpha_{G},S_G) + \widehat{p}(S_B,\alpha_{G},S_G)   \right)(C_H-C_L)}{I_p(\alpha_{G},A_H)I_p(\alpha_{G},A_L)}.
\end{array}
\end{equation}

As a result, we have
\begin{equation}
\label{eq:UserOVDiffG}
\begin{array}{l}
\overline{V}(S_G,A_H;\alpha_{B},R)  - \overline{V}(S_G,A_L;\alpha_{B},R)
\\ = (1-R) \left( k(S_G,A_H;\alpha_{B})  - k(S_G,A_L;\alpha_{B})    \right)   + b(S_G,A_H;\alpha_{B})  - b(S_G,A_L;\alpha_{B})
\\ = \frac{(1-\delta)\widehat{p}(S_B,\alpha_{B},S_B) }{I_p(A_H,\alpha_{B}) I_p(A_L,\alpha_{B})}  h(S_G,\alpha_B,R);
\end{array}
\end{equation}
\begin{equation}
\label{eq:UserOVDiffB}
\begin{array}{l}
\overline{V}(S_B,A_H;\alpha_G,R) -\overline{V}(S_B,A_L;\alpha_G,R)   = \frac{(1-\delta)\widehat{p}(S_G,\alpha_{G},S_G) }{I_p(\alpha_G,A_H) I_p(\alpha_G,A_L)}  
h(S_B,\alpha_G,R),
\end{array}
\end{equation}
where $h(s,\alpha_{s^c}, R)$ has been defined in Proposition \ref{pro:UserH}. Since $1-\delta > 0$, $\widehat{p}(s, \alpha_s, s) > 0$, and $I_p(\alpha_s, \alpha_{s^c}) > 0$, we have that $\overline{V}(s,A_H;\alpha_{s^c},R) < \overline{V}(s,A_L;\alpha_{s^c},R)$ if $h(s,\alpha_{s^c},R)<0$ and $\overline{V}(s,A_H;\alpha_{s^c},R) \geq  \overline{V}(s,A_L;\alpha_{s^c},R)$ if $h(s,\alpha_{s^c},R)\geq 0$. Proposition \ref{pro:UserH} holds. 

\section*{Appendix B. Proof of Theorem \ref{the:UserUniqueness}}
Recall the value of transition probabilities $\rho$ from (\ref{eq:UserValueTransitionProbabilities1}) in Proposition \ref{pro:UserSwitchingCoverageLevel}. Besides Proposition \ref{pro:UserHMonotonicity}, we note that $h(s,\alpha_{s^c},R)$ has the following extra facts.
\begin{equation}
\label{eq:AppInvariability}
\begin{array}{l}
h(S_G,A_H,R) - h(S_G,A_L,R) =  h(S_B,A_H,R) - h(S_B,A_L,R)  =\rho \delta( C_H-C_L );
\end{array}
\end{equation}
\begin{equation}
\label{eq:AppMonotonicityI}
\begin{array}{l}
h(S_G,A_H,R) -  h(S_B,A_H,R)  =  h(S_G,A_L,R) -  h(S_B,A_L,R)   = \rho \delta (1-R)  ( X_B - X_G  ) ;
\end{array}
\end{equation}
\begin{equation}
\label{eq:AppMonotonicityII}
\begin{array}{l}
h(S_G,A_H,R) -  h(S_B,A_L,R) 
  = \rho\delta \left(  ( C_H - C_L  )  +    (1-R) ( X_B - X_G  )     \right)   ;
\end{array} 
\end{equation}
\begin{equation}
\label{eq:AppVariability}
\begin{array}{l}
h(S_B,A_H,R) -  h(S_G,A_L,R) =  \rho \delta \left(   ( C_H - C_L  )  -    (1-R) ( X_B - X_G  )     \right)  .
\end{array}  
\end{equation}

If $\pi_R^* = \Pi_{HL}$, we have $h(S_G,A_L,R) < 0$ and $h(S_B,A_H,R) \geq 0$ from Proposition \ref{pro:UserH}. Thus, $\pi_R^* \neq \Pi_{LL}$ and $\pi_R^* \neq \Pi_{HH}$. Similarly, if $\pi_R^* = \Pi_{LH}$, we have $\pi_R^* \neq \Pi_{LL}$ and $\pi_R^* \neq \Pi_{HH}$; if $\pi_R^* = \Pi_{LL}$, we have $\pi_R^* \neq \Pi_{LH}$ and $\pi_R^* \neq \Pi_{HL}$; if $\pi_R^* = \Pi_{HH}$, we have $\pi_R^* \neq \Pi_{LH}$ and $\pi_R^* \neq \Pi_{HL}$. Thus, to prove the uniqueness of $\pi_R^*$, we only need to prove that (i). $\pi_R^* = \Pi_{LH}$ and $\pi_R^* = \Pi_{HL}$ cannot exist at the same time and (ii). $\pi_R^* = \Pi_{LL}$ and $\pi_R^* = \Pi_{HH}$ cannot exist at the same time. 

If $\pi_R^* = \Pi_{LH}$ and $\pi_R^* = \Pi_{HL}$ at the same time, we have $h(S_G,A_H,R) \geq 0$, $h(S_B,A_L, R) < 0$, $h(S_G,A_L,R) < 0$, and $h(S_B,A_H,R) \geq 0$ from Proposition \ref{pro:UserH}, which indicates that $\rho > 0$ from (\ref{eq:AppMonotonicityII}) as $h(S_G,A_H,R)> h(S_B,A_L, R)$ and $\rho \delta \left(   ( C_H - C_L  )  -    (1-R) ( X_B - X_G  )     \right)  > 0$ from (\ref{eq:AppVariability}) as $h(S_G,A_L,R) < h(S_B,A_H,R)$. Thus, we can achieve that $( C_H - C_L ) > (1-R) ( X_B - X_G )$. However, 
\begin{equation}
\label{eq:AppHGLR}
\begin{array}{l}
h(S_G,A_L, R)  \\= (1-R)\delta  ( p(S_G,A_H,S_B)   -   p(S_G,A_L,S_B) ) ( X_B - X_G) \\ \ \ \ \ \ \ \ \ \ \ \ \ \ \ \ \     +   (1-\delta + \delta p(S_B,A_L,S_G)   +   \delta p(S_G,A_L,S_B)  )( C_H - C_L  )

\\ >  (1-R) \delta ( p(S_G,A_H,S_B)   -   p(S_G,A_L,S_B) ) (X_B - X_G)  \\ \ \ \ \ \ \ \ \ \ \ \ \ \ \ \ \ +  \left(   1-\delta + \delta p(S_B,A_L,S_G)   +   \delta p(S_G,A_L,S_B)  \right)(1-R)( X_B - X_G)

\\ =  (1-R) ( X_B - X_G) \left(1-\delta +   \delta p(S_G,A_H,S_B) + \delta   p(S_B,A_L,S_G)  \right)

\\ > 0, 
\end{array}
\end{equation}
which violates $h(S_G,A_L, R) < 0$. As a result, $\pi_R^* = \Pi_{LH}$ and $\pi_R^* = \Pi_{HL}$ cannot exist at the same time.

If $\pi_R^* = \Pi_{LL}$ and $\pi_R^* = \Pi_{HH}$ at the same time, we have $h(S_G,A_L,R) \geq 0$, $h(S_B,A_L,R) \geq 0$, $h(S_G,A_H,R) < 0$, and $h(S_B,A_H,R) < 0$, which indicates that $\rho < 0$ from (\ref{eq:AppInvariability}) and $\rho \delta \left(   ( C_H - C_L  )  -    (1-R) ( X_B - X_G  )     \right)  < 0$ from (\ref{eq:AppVariability}). Thus, we can achieve that $( C_H - C_L )  -    (1-R) ( X_B - X_G )   > 0$. However, 
\[\begin{array}{l}
h(S_B,A_H, R)  \\= (1-R)\delta  ( p(S_B,A_H,S_G)   -   p(S_B,A_L,S_G) ) ( X_G - X_B) \\ \ \ \ \ \ \ \ \ \ \ \ \ \ \ \ \       +    (1-\delta + \delta p(S_B,A_H,S_G)   +   \delta p(S_G,A_H,S_B)  )( C_H - C_L  ) 

\\ >  (1-R) \delta ( p(S_B,A_H,S_G)   -   p(S_B,A_L,S_G) ) (X_G - X_B) \\ \ \ \ \ \ \ \ \ \ \ \ \ \ \ \ \ +  \left(   1-\delta + \delta p(S_B,A_H,S_G)   +   \delta p(S_G,A_H,S_B)  \right)(1-R)( X_B - X_G)

\\ =  (1-R)( X_B - X_G) \left(   1-\delta  +   \delta p(S_G,A_H,S_B) + \delta  p(S_B,A_L,S_G)  \right)

\\ > 0,
\end{array}\]
which violates $h(S_B,A_H, R) < 0$. As a result, $\pi_R^*=\Pi_{LL}$ and $\pi_R^*=\Pi_{HH}$ cannot exist at the same time. Thus, Theorem \ref{the:UserUniqueness} holds.

\section*{Appendix C. Proof of Proposition \ref{pro:UserSwitchingCoverageLevel}}
There are only four possible protection policies $\Pi_{LL}$, $\Pi_{HL}$, $\Pi_{LH}$, and $\Pi_{HH}$. Thus, the optimal protection policy $\pi_0^*$ without insurance has only four cases: Case 1, Case 2, Case 3, and Case 4 as presented in Proposition \ref{pro:UserSwitchingCoverageLevel}, which are determined by $h(s, \alpha_{s^c}, 0)$ in Proposition \ref{pro:UserH}. As a result, we only need to prove the trends of $\pi_R^*$ with respect to $R$ in different cases. 

We first note that when $R = 1$, we have $h(S_G,\alpha_{B},1)> 0$ and $h(S_B,\alpha_{G},1)> 0$, which indicates that $\pi_{R=1}^* = \Pi_{LL}$. Moreover, if the user has $\pi_{\widehat{R}}^* = \Pi_{LL}$ for a coverage level of $\widehat{R}\in [0, 1]$, we have $h(S_G,A_L,\widehat{R}) \geq 0$ and $h(S_B,A_L,\widehat{R}) \geq 0$ from Proposition \ref{pro:UserH}. Since $h(s,\alpha_{s^c},R)$ is linearly increasing on $R$ as shown in Proposition \ref{pro:UserHMonotonicity}, we have $h(S_G,A_L,R) \geq 0$ and $h(S_B,A_L,R) \geq 0$ for $R \geq \widehat{R}$, which indicates that $\pi_{R}^* = \Pi_{LL}$ for $R \geq \widehat{R}$. Thus, we can conclude that $\pi_R^* = \Pi_{LL}$ when $R$ is sufficiently large and the user chooses not to change his policy with the increase of $R$ once he achieves $\pi_R^*=\Pi_{LL}$ for all cases. 

To prove Cases 2, 3, and 4, recall the value of transition probabilities $\rho$ from (\ref{eq:UserValueTransitionProbabilities1}). If $\rho < 0$, we have $h(S_G,A_H, R) < h(S_B,A_L, R)$ from (\ref{eq:AppMonotonicityII}). However, when $\pi_R^* = \Pi_{LH}$, we have $h(S_G,A_H,R) \geq 0$ and $h(S_B,A_L, R)    < 0$ from Proposition \ref{pro:UserH}, which violates $h(S_G,A_H, R) < h(S_B,A_L, R)$. Thus, we have $\pi_R^*\neq \Pi_{LH}$ if $\rho < 0$. As a result, Case 2 and Case 4(a) has $\rho < 0$ and the user have $\pi_R^* \neq \Pi_{LH}$ in these cases. Since $h(S_B,A_H,0) \geq 0$ when $\pi_{0}^* = \Pi_{HL}$ in Case 2 and $h(S_B,A_H,R)$ is linearly increasing on $R$, $h(S_B,A_H,R) \geq 0$ for $R \in [0,1]$. Thus, the user has $\pi_R^* \neq \pi_{HH}$ in Case 2 and Case 2 holds. The threshold $R_G$ is achieved by solving $h(S_G, A_L, R) = 0$. Case 4(a) holds from Case 2, and the thresholds $R_G$ and $R_B$ are achieved by solving $h(S_G, A_L, R) = 0$ and $h(S_B, A_H, R) = 0$, respectively. 

If $\rho > 0$ and $\pi_R^* = \Pi_{HL}$, we have $h(S_G,A_L,R) < 0$ and $h(S_B,A_H,R) \geq 0$ from Proposition \ref{pro:UserH} and thus $h(S_B,A_H, R) -  h(S_G,A_L, R) = \rho \delta \left(   ( C_H - C_L  )  -    (1-R) ( X_B - X_G  )     \right)  > 0$ from (\ref{eq:AppVariability}), which indicates that $( C_H - C_L )  -    (1-R) ( X_B - X_G )   > 0$. However, we could obtain that $h(S_G,A_L, R) > 0$ following similar arguments as in (\ref{eq:AppHGLR}), which violates $h(S_G,A_L, R) < 0$. Thus, we have $\pi_R^*\neq \Pi_{HL}$ if $\rho > 0$. As a result, Case 3 and Case 4(b) have $\rho > 0$ and the user has $\pi_R^* \neq \Pi_{HL}$ in these cases. Since $h(S_G,A_H,R) \geq 0$ when $\pi_{0}^* = \Pi_{LH}$ in Case 3 and $h(S_G,A_H,R)$ is linearly increasing on $R$,  $h(S_G,A_H,R) \geq 0$ for $R \in [0,1]$. Thus, the user has $\pi_R^* \neq \pi_{HH}$ in Case 3 and Case 3 holds. The threshold $R_B$ is achieved by solving $h(S_B, A_L, R) = 0$. Case 4(b) holds from Case 3, and the thresholds $R_B$ and $R_G$ are achieved by solving $h(S_B, A_L, R) = 0$ and $h(S_G, A_H, R) = 0$, respectively. 

If $\rho = 0$, we have $h(S_G,A_H, R)= h(S_G,A_L,R)  = h(S_B,A_H, R) = h (S_B,A_L, R)$ from (\ref{eq:AppInvariability})-(\ref{eq:AppVariability}). However, $\Pi_{LH}$ and $\Pi_{HL}$ indicate that $ h(S_G,A_H,R)  \geq  0 >   h (S_B,A_L, R) $  and $h(S_G,A_L,R) <0 \leq h(S_B,A_H, R) $, respectively. Thus, we have $\pi_R^*\neq \Pi_{HH}$ and $\pi_R^*\neq \Pi_{LL}$ if $\rho = 0$. As a result, Case 4(c) has $\rho = 0$ and the user has $\pi_R^* \neq \Pi_{HL}$ and  $\pi_R^* \neq \Pi_{LH}$. Thus, Case 4(c) holds, and the thresholds $R_G$ and $R_B$ are achieved by solving $h(S_G, A_H, R) = 0$ and $h(S_B, A_H, R) = 0$, respectively.

\bibliographystyle{ieeetr}
\bibliography{elsarticle-template.bib}

\end{document}